\documentclass[runningheads, envcountsame, a4paper]{llncs}
\usepackage[T1]{fontenc}
\usepackage{graphicx}
\usepackage{amsmath}
\usepackage{amssymb}
\usepackage{amsfonts}
\usepackage[hidelinks]{hyperref}
\usepackage[toc,page]{appendix}
\usepackage{subcaption}
\usepackage[utf8]{inputenc}
\usepackage[
    giveninits=true,
    backend=biber,
    url=false,
    isbn=false,
    doi=false,
    maxbibnames=20
]{biblatex}
\bibliography{bibliography.bib}
\AtEveryBibitem{%
    \clearname{editor}%
    \clearfield{copyright}%
    \clearlist{publisher}%
    \clearlist{location}%
    \clearfield{month}%
    \clearfield{pages}%
}
\usepackage[misc,geometry]{ifsym}
\usepackage{orcidlink}

\usepackage[linesnumbered,ruled,vlined]{algorithm2e}
\newcommand{\weakenuntildirect}{\ensuremath{Weaken\mathcal{U}Direct}}
\newcommand{\weakenreleasedirect}{\ensuremath{Weaken\mathcal{R}Direct}}
\newcommand{\weakenuntilleft}{\ensuremath{Weaken\mathcal{U}Left}}
\newcommand{\weakenuntilright}{\ensuremath{Weaken\mathcal{U}Right}}
\newcommand{\weakenreleaseleft}{\ensuremath{Weaken\mathcal{R}Left}}
\newcommand{\weakenreleaseright}{\ensuremath{Weaken\mathcal{R}Right}}
\newcommand{\weakenaux}{\ensuremath{Weaken}}

\usepackage{color}

\usepackage{paralist}
\usepackage{mathtools}
\usepackage{listingsutf8}
\usepackage{pmboxdraw}
\protected\def\ucrightarrow{\ensuremath{\rightarrow}}
\usepackage{tikz}
\usetikzlibrary{calc}
\usetikzlibrary{patterns}
\usetikzlibrary{arrows.meta, positioning}
\tikzset{
    state/.style={
        draw=black,
        shape=rectangle,
        minimum width=2cm,
        minimum height=1cm,
        text=black,
    },
    cex/.style={
        draw=black,
        shape=rectangle,
        minimum height=0.75cm,
        text=black,
    },
    >=Stealth,
    every edge/.append style={draw=black,->,font=\small,text=black}
}
\usepackage{pgfplots}
\pgfplotsset{compat=1.18}
\usepackage[capitalise]{cleveref}
\usepackage{xcolor}
\definecolor{fret_scope}{HTML}{9F0500}
\definecolor{fret_condition}{HTML}{FB9E00}
\definecolor{fret_component}{HTML}{68BC00}
\definecolor{fret_shall}{HTML}{000000}
\definecolor{fret_probability}{HTML}{009CE0}
\definecolor{fret_timing}{HTML}{0062B1}
\definecolor{fret_response}{HTML}{653294}

\newcommand{\bnfsep}{\;|\;}
\newcommand{\until}[1]{\;\mathcal{U}_{#1}\,}

\newcommand{\release}[1]{\;\mathcal{R}_{#1}\,}
\newcommand{\bintemporal}[1]{\,\triangle_{#1}}
\newcommand{\generally}[1]{\Box_{#1}}
\newcommand{\eventually}[1]{\Diamond_{#1}}
\newcommand{\N}{\mathbb{N}}

\newcommand{\model}{\mathcal{M}}
\newcommand{\wkn}{\sqsubseteq}
\newcommand{\MTL}{\textsf{MTL}}
\newcommand{\LTL}{\textsf{LTL}}
\newcommand{\Iorig}{I_{\mathrm{orig}}}
\newcommand{\pipre}{\pi_{\mathrm{pre}}}
\newcommand{\pisuf}{\pi_{\mathrm{suf}}}
\newcommand{\bfin}{b_\mathrm{fin}}

\SetKwProg{Function}{\textbf{function} }{}{}
\DeclareMathOperator{\rightidx}{end}
\DeclareMathOperator{\covering}{cov}
\newcommand{\rightmodifications}{\mathcal{B}_R}
\newcommand{\sat}[2]{#1,#2\vDash}
\newcommand{\nsat}[2]{#1,#2\nvDash}

\newcommand{\fret}{FRET}
\newcommand{\fretish}{\textsc{fretish}}

\newcommand{\nuXmv}{\textsc{nuXmv}}
\newcommand{\SMV}{\textsc{SMV}}
\newcommand{\spin}{\textsc{Spin}}
\newcommand{\cegiw}{\textsc{CEGIW}}

\makeatletter
\newcommand{\centereqn}[1]{%
  \refstepcounter{equation}%
  \noindent
  \makebox[\textwidth]{%
    \hfill
    \makebox[0pt][c]{$#1$}
    \hfill\tagform@{\theequation}%
  }%
}
\makeatother

\expandafter\def\expandafter\normalsize\expandafter{%
    \normalsize%
    \setlength\abovedisplayskip{4pt}%
    \setlength\belowdisplayskip{4pt}%
    \setlength\abovedisplayshortskip{4pt}%
    \setlength\belowdisplayshortskip{4pt}%
}

\begin{document}

\title{Counterexample-Guided Interval Weakening\thanks{This work is partially supported by EPSRC grant EP/Y001532/1, the CRADLE
project under EPSRC grant EP/X02489X/1 and the Royal Academy of Engineering
through both a Research Fellowship and a Chair in Emerging Technology}}

\author{Ben M.\ Andrew$^{(\textrm{\Letter})}$\orcidlink{0009-0009-8910-5899} \and
Louise A.\ Dennis\orcidlink{0000-0003-1426-1896} \and
Michael Fisher\orcidlink{0000-0002-0875-3862} \and 
Marie Farrell\orcidlink{0000-0001-7708-3877}
}

\authorrunning{B.M.\ Andrew et al.}

\institute{Department of Computer Science, University of Manchester, UK\\
\email{benjamin.andrew@manchester.ac.uk}}

\maketitle

\begin{abstract}
Systems deployed for long periods of time in dynamic environments may experience performance degradation that affects timing guarantees, even when their functional behaviour remains unchanged. In the design and verification of critical systems, such timing guarantees are often expressed using Metric Temporal Logic (\MTL{}). Under degradation, these specifications may no longer hold as stated, although weaker variants that relax timing bounds may still be satisfied and remain meaningful. For example, while an elevator may initially be required to arrive within 30 seconds of a request, degradation of its motor may only allow us to guarantee arrival within 60 seconds. Although weaker, this guarantee is still useful and allows the system to maintain a reasonable level of operation. In this paper we present \cegiw{}, an iterative, counterexample-guided algorithm for automatically weakening timing intervals in \MTL{} specifications so that they hold for a given system model. The algorithm preserves the logical structure of the original specification and weakens only interval bounds. We prove the correctness and optimality of \cegiw{}, and conduct an empirical evaluation to demonstrate the practicality of interval weakening using formalised requirements from a number of real-world case-studies. Using a model checker to produce counterexamples, \cegiw{} either identifies the strongest interval weakening under which the specification holds, or determines that no such weakening exists.

\keywords{System degradation \and Specification weakening \and Formal \\methods \and Metric temporal logic}
\end{abstract}

\section{Introduction}

Temporal properties of systems are often specified using logics such as Metric Temporal Logic~\cite{koymans1990} (\MTL{}), and these properties can be verified to hold using model-checking~\cite{cavada2014}. However, in the real world, system failures or degradation can invalidate these proofs by breaking their assumptions, in which case the desired properties may no longer hold. Yet, under degradation the system may still have some useful capabilities for reduced operation, and so \emph{logically weaker} versions of these properties may hold.
Given a degraded system and an ideal \MTL{} property that does not hold in it, we aim to derive the strongest possible logical \emph{weakening} of that property. To constrain the search space, we focus on modifying the intervals of \MTL{} formulae while preserving the structural form of the specification. For example, we may want an elevator to always arrive at least 30 seconds after calling it, represented by
\begin{equation}
\Box(\texttt{callElevator}\rightarrow\eventually{[0,30]}\texttt{elevatorArrives})
\end{equation}
(where $\Box$ is the \textit{always} operator and $\Diamond_{[0,30]}$ is the \textit{eventually} operator bounded between zero and thirty time units). However, if the main motor breaks, a weaker backup motor may start, slowing the system down. In this case the ideal property may not hold, and we may only be able to guarantee that the elevator will arrive within 60 seconds, represented by
\begin{equation}
\Box(\texttt{callElevator}\rightarrow\eventually{[0,60]}\texttt{elevatorArrives}).
\end{equation}
This property is logically weaker than the original, but still guarantees a useful level of functionality. We would like to be able to derive this new property automatically from the system model and the original property.

\textbf{Related work.}
Many works consider \emph{unrealisable} sets of requirements --- where conflicts mean that no satisfying implementation exists --- solving the problem by weakening specifications. Some use counterstrategies to strengthen assumptions~\cite{alur2013a,maoz2019} in the \LTL{} fragment $GR(1)$, while others use heuristic-guided genetic algorithms to mutate assumptions and guarantees towards realisability~\cite{brizzio2023}. However, this is different from the problem of weakening specifications relative to an existing implementation, which we are concerned with. We use a counterexample-guided approach, which has been applied to a large variety of problems including abstraction refinement~\cite{clarke2000,aarts2012,howar2011}, program synthesis~\cite{alur2013}, and learning assumptions for compositional verification~\cite{cobleigh2003}, but not yet to the problem of specification repair in the presence of an existing implementation. This has been explored using techniques from the field of program repair~\cite{gazzola2018}, typically heuristically-guided \emph{generate-and-validate} approaches like mutation-based repairs~\cite{cerqueira2022} and dynamic invariant detection~\cite{abreu2023a}. We, however, are concerned with correct-by-construction, \emph{semantics-driven} approaches, which have only been explored in the case of propositional logic specifications~\cite{andrew2026}.

\textbf{Contribution.}
We present our Counterexample-Guided Interval Weakening (\cegiw) algorithm that, given a degraded system and a desired \MTL{} property that does not hold on the system, produces a new optimal \MTL{} property that both is weakening of the original property and holds in the degraded system. The weakening is optimal with respect to a formally defined interval order, ensuring that no strictly stronger interval weakening satisfies the degraded system. We use a counterexample-guided approach, generating counterexamples with the \nuXmv{} model checker~\cite{cavada2014}, weakening the property to hold on the counterexamples, and iteratively weakening in this way until the property holds in the system. This approach is aimed at engineers in the design phase of safety-critical systems, who are trying to understand how resilient the timing properties of their system are to various proposed degradations, and how the system's formal guarantees are thus impacted.

The paper is organised as follows: \cref{sec:contexts} sets up the weakening of \MTL{} formulae within contexts, \cref{sec:algorithm} describes \cegiw{} and proves its correctness and optimality, \cref{sec:evaluation} demonstrates \cegiw{} on an example and considers its usefulness in real-world case-studies, and \cref{sec:conclusion} concludes and outlines future work.

\section{Weakening Within Contexts}\label{sec:contexts}

We briefly state the syntax and semantics of Metric Temporal Logic~\cite{koymans1990} (\MTL{}). Let $\mathcal{P}$ be a set of propositional variables. Well-formed \MTL{} formulae are formed according to the rule:
\begin{equation}
    \phi := p \bnfsep \top \bnfsep \neg \phi \bnfsep \phi \land \phi \bnfsep \phi \until{I} \phi \bnfsep \phi \release{I} \phi
\end{equation}
where $p\in\mathcal{P}$ and $I$ is an interval, $[a,b]$, for $a\in\N$ and $b\in\N\cup\{\infty\}$ and $a\leq b$. Other constructs can be defined as usual, e.g. $\eventually{I}\phi\equiv\top\until{I}\phi$. We consider \MTL{} formulae with a pointwise semantics over the set of natural numbers~\cite{alur1993}, defined according to a trace $\pi$ which is an infinite sequence of states in which atomic propositions can hold, and an index of the trace $t \in \N$. The set of atomic propositions that hold in the $t$-th state is denoted by $\pi(t)$. A trace $\pi$ satisfies an \MTL{} formula $\phi$, denoted by $\pi\vDash\phi$, if and only if $\sat{\pi}{0}\phi$.
\begin{align*}
    \pi, t & \vDash p && \text{iff}&& p \in \pi(t) \\
    \pi, t & \vDash \neg \phi && \text{iff}&& \pi, t \nvDash \phi \\
    \pi, t & \vDash \phi_1 \land \phi_2 && \text{iff}&& \pi, t \vDash \phi_1 \text{ and } \pi, t \vDash \phi_2 \\
    \pi, t & \vDash \phi_1 \until{I} \phi_2 && \text{iff}&& \exists i \in I.\, ((\pi, t+i \vDash \phi_2) \land \forall j \in [0,i) \cap I.\, (\pi, t+j \vDash \phi_1)) \\
    \pi, t & \vDash \phi_1 \release{I} \phi_2 && \text{iff}&& \forall i \in I.\, (\pi, t+i \vDash \phi_1) \\
    & && && \lor \exists j \in I.\, (\pi, t+j \vDash \phi_2 \land \forall i \in [0,i] \cap I.\, (\pi, t+i \vDash \phi_1))
\end{align*}
Note that $\phi_1\release{I}\phi_2\equiv\neg(\neg\phi_1\until{I}\neg\phi_2)$. \cegiw{} weakens a constituent subformula of a larger formula. We show, using the notion of \emph{contexts}, that a weakening of a subformula implies a weakening of the larger formula.
\begin{definition}[Contexts]
\MTL{} Contexts are like \MTL{} formulae with a single hole $[-]$, and are formed according to the rule:
\begin{equation}
C ::= [-] \bnfsep C \land \phi \bnfsep \phi \land C \bnfsep C \lor \phi \bnfsep \phi \lor C \bnfsep C \until{I} \phi \bnfsep \phi \until{I} C \bnfsep C \release{I} \phi \bnfsep \phi \release{I} C
\end{equation}
where $\phi$ is an \MTL{} formula and $I$ is an interval. Our definition of contexts does not allow negations on the path to the hole $[-]$, similarly to the restriction imposed by negation normal form (NNF). However, adjacent \MTL{} subformulae $\phi$ are not required to be in NNF and can contain negations.
\end{definition}
We define the notion of \emph{context substitution}, where an \MTL{} formula $\psi$ is substituted into the hole of a context $C$ to produce an \MTL{} formula $C[\psi]$. By pushing negations inwards, an arbitrary subformula $\psi$ of an \MTL{} formula $\phi$ can always be extracted to get a context $C$ where $\phi$ is logically equivalent to $C[\psi]$.

\makeatletter
\refstepcounter{equation}%
\noindent\begin{minipage}{\textwidth}
\begin{minipage}{.5\linewidth}
\begin{align*}
[-][\psi] &= \psi \\
(C\land\phi)[\psi] &= C[\psi]\land\phi \\
(\phi\land C)[\psi] &= \phi\land C[\psi] \\
(C\lor\phi)[\psi] &= C[\psi]\lor\phi \\
(\phi\lor C)[\psi] &= \phi\lor C[\psi]
\end{align*}
\end{minipage}%
\begin{minipage}{.5\linewidth}
\begin{align*}
\\
(C\until{I}\phi)[\psi] &= C[\psi]\until{I}\phi \\
(\phi\until{I} C)[\psi] &= \phi\until{I} C[\psi] \\
(C\release{I}\phi)[\psi] &= C[\psi]\release{I}\phi \\
(\phi\release{I} C)[\psi] &= \phi\release{I} C[\psi]
\end{align*}
\end{minipage}%
\makebox[0pt][r]{\tagform@{\theequation}}%
\label{eqn:my-equation}%
\end{minipage}
\makeatother
\vspace{\baselineskip}

\begin{definition}[Weakening and strengthening of \MTL{} formulae]
    Let $\phi$ and $\phi'$ be \MTL{} formulae. $\phi'$ is a weakening of $\phi$, denoted
    \begin{equation}
    \phi\wkn\phi'
    \end{equation}
    \noindent if and only if, for all traces $\pi$ and time-points $t$, if $\sat{\pi}{t}\phi$, then $\sat{\pi}{t}\phi'$. In this case, symmetrically, $\phi$ is a strengthening of $\phi'$. Note that an \MTL{} formula $\phi$ is always both a strengthening and a weakening of itself, i.e. $\phi\wkn\phi$.
\end{definition}

\begin{theorem}[Weakening of contexts]\label{theorem:context-wkn}
Let $C$ be a context and $\psi$ and $\psi'$ be \MTL{} formulae. If $\psi\wkn\psi'$, then $C[\psi]\wkn C[\psi']$.
\end{theorem}

\begin{proof}
We do a proof by induction over the grammar of contexts using the induction hypothesis $P(C)$, that $C[\psi]\wkn C[\psi']$. The non-temporal inductive cases are omitted for brevity.
\begin{description}
\item \textsc{Base case $[-]$:} We assume that $\psi\wkn\psi'$, and by the definition of context substitution we have that $[-][\psi]\wkn[-][\psi']$ and thus $P([-])$.

\item \textsc{Inductive case $C\until{I}\phi$:} Assuming $P(C)$, we take an arbitrary trace $\pi$ and time-point $t$, assume $\pi,t\vDash C[\psi]\until{I}\phi$, and want to prove $\pi,t\vDash C[\psi']\until{I}\phi$. We know that there exists an $i\in I$ such that $\pi,t+i\vDash\phi$, and that for all $j\in[0,i)\cap I$ we have $\pi,t+j\vDash C[\psi]$. Taking arbitrary $i$ and $j$, by the induction hypothesis we have that $\pi,t+j\vDash C[\psi']$, and so by the semantics $\pi,t\vDash C[\psi']\until{I}\phi$. Thus, we have $P(C\until{I}\phi)$.

\item \textsc{Inductive case $\phi\until{I}C$:} Similar to the above case.

\item \textsc{Inductive case $C\release{I}\phi$:} Assuming $P(C)$, we take an arbitrary trace $\pi$ and time-point $t$, assume $\pi,t\vDash C[\psi]\release{I}\phi$, and want to prove $\pi,t\vDash C[\psi']\release{I}\phi$. By the semantics of $\mathcal{R}$ there are two cases:
\begin{enumerate}
    \item For all $i\in I$ we have $\pi,t+i\vDash\phi$, thus we have $\pi,t+i\vDash C[\psi']$, and so we have $\pi,t\vDash C[\psi']\release{I}\phi$.

    \item There exists an $i\in I$ such that $\pi,t+i\vDash C[\psi]$, and that for all $j\in[0,i]\cap I$ we have $\pi,t+j\vDash\phi$. Taking arbitrary $i$ and $j$, by the assumptions we have that $\pi,t+i\vDash C[\psi']$ and $\pi,t+j\vDash\phi$, and then by the semantics we have $\pi,t\vDash C[\psi']\release{I}\phi$.
\end{enumerate}
Thus, in both cases we have $P(C\release{I}\phi)$.

\item \textsc{Inductive case $\phi\release{I}C$:} Similar to the above case.\qed
\end{description}
\end{proof}

We show that, depending on which temporal operator is used, by expanding or contracting its interval we can weaken or strengthen the surrounding formula.

\begin{definition}[Right-bound modifications of intervals]\label{def:interval-extension}
    Let $I=[a,b]$ be an interval. For any $i\in\N$, a right-bound modification of $I$ is either a right-bound extension $[a,b+i]$, or, provided $i\leq b - a$, a right-bound contraction $[a,b-i]$. A right-bound modification is \emph{strict} if $i>0$. The set of all right-bound modifications of $I$ is denoted $\rightmodifications(I)$.
\end{definition}

\begin{lemma}[Weakening of $\mathcal{U}$ interval]\label{lemma:until-i}
    Let $\phi$ and $\psi$ be \MTL{} formulae, and $I$ and $I'$ be intervals, where $I'$ is a right-bound extension of $I$. Then, $\phi\until{I}\psi\wkn\phi\until{I'}\psi$.
\end{lemma}

\begin{proof}
    We assume that $I'$ is a right-bound extension of $I$, and so, taking an arbitrary trace $\pi$ and time-point $t$, we assume $\pi,t\vDash \phi\until{I}\psi$ and want to prove $\pi,t\vDash \phi\until{I'}\psi$. We know that there exists an $i\in I$ such that $\pi,t+i\vDash\psi$, and that for all $j\in [0,i)\cap I$ we have $\pi,t+j\vDash\phi$. Taking arbitrary $i$ and $j$, we have that $i,j\in I'$, and so $\pi,t\vDash \phi\until{I'}\psi$. Thus, $\phi\until{I}\psi\wkn\phi\until{I'}\psi$.
    \qed
\end{proof}

\begin{lemma}[Weakening of $\mathcal{R}$ interval]\label{lemma:release-i}
    Let $\phi$ and $\psi$ be \MTL{} formulae, and $I$ and $I'$ be intervals, where $I'$ is a right-bound contraction of $I$. Then, $\phi\release{I}\psi\wkn\phi\release{I'}\psi$.
\end{lemma}

\begin{proof}
    We assume that $I'$ is a right-bound contraction of $I$, and so, taking an arbitrary trace $\pi$ and time-point $t$, we assume $\pi,t\vDash \phi\release{I}\psi$ and want to prove $\pi,t\vDash \phi\release{I'}\psi$. By the semantics of $\mathcal{R}$ there are two cases:
    \begin{enumerate}
        \item For all $t'\in I$ we have $\pi,t+t'\vDash\psi$. Then, as $I'\subseteq I$, we know that for all $t''\in I'$ we have $\pi,t+t''\vDash\psi$, and so $\pi,t\vDash \phi\release{I'}\psi$.
        
        \item There exists a $t'\in I$ such that $\pi,t+t'\vDash\phi$ and for all $t''\in I\cap[0,t']$, we have $\pi,t+t''\vDash\psi$. As $I'$ is a right-bound contraction of $I$, there are two further cases:
        
        \begin{enumerate}
            \item If $t'\in I'$, then we still have that $\pi,t+t'\vDash\phi$ and for all $t''\in I'\cap[0,t']$, we have $\pi,t+t''\vDash\psi$, and so $\pi,t\vDash \phi\release{I'}\psi$.
    
            \item If $t'\notin I'$, then $I'\cap[0,t']=I'$ and so we know that for all $t''\in I'$, we have $\pi,t+t''\vDash\psi$, and so $\pi,t\vDash \phi\release{I'}\psi$.
        \end{enumerate}
    \end{enumerate}
    Thus, in all cases we have that $\phi\release{I}\psi\wkn\phi\release{I'}\psi$.
    \qed
\end{proof}

Often in \cegiw{}, recursive calls will generate a set of intervals from which either the strongest or weakest must be chosen. We show that there is a total order of implication over the set of right-bound modifications of an interval, which allows us to make that choice.

\begin{lemma}[Extension-weakening order of right-bound modifications]\label{lemma:extension-weakening-order}
    Let $I$ be an interval. Then, $\rightmodifications(I)$ has a total order $\supseteq$, where for all \MTL{} contexts $C$, \MTL{} formulae $\phi$ and $\phi'$, and all $I',I''\in\rightmodifications(I)$, if $I''\supseteq I'$ then we have $C[\phi\until{I'}\phi']\wkn C[\phi\until{I''}\phi']$.
\end{lemma}

\begin{proof}
    For any pair of intervals $I'$ and $I''$ in $\rightmodifications(I)$ we can order the resulting subformulae by applying \cref{lemma:until-i} to get $\phi\until{I'}\phi'\wkn\phi\until{I''}\phi'$ (or the reverse), and then order the full formulae with their contexts by applying \cref{theorem:context-wkn} to get $C[\phi\until{I'}\phi']\wkn C[\phi\until{I''}\phi']$ (or the reverse).
    \qed
\end{proof}

\begin{lemma}[Contraction-weakening order of right-bound modifications]\label{lemma:contraction-weakening-order}
    Let $I$ be an interval. Then, $\rightmodifications(I)$ has a total order $\subseteq$, where for all \MTL{} contexts $C$, \MTL{} formulae $\phi$ and $\phi'$, and all $I',I''\in\rightmodifications(I)$, if $I''\subseteq I'$ then we have $C[\phi\release{I'}\phi']\wkn C[\phi\release{I''}\phi']$.

\end{lemma}

\begin{proof}
    Similar to the proof of \cref{lemma:extension-weakening-order}, but uses \cref{lemma:release-i} to order $\phi\release{I'}\phi'$ and $\phi\release{I''}\phi'$.\qed
\end{proof}

\section{Algorithm for Interval Weakening}\label{sec:algorithm}

\cegiw{} is split into two levels. At the top-level, there is an iterative process that finds counterexample traces to the \MTL{} formula $\phi$ by model checking (\cref{subsec:iterative-weakening}). At each iteration, once a counterexample trace $\pi$ is found, we weaken a given interval in $\phi$ such that the new formula $\phi'$ holds on $\pi$ (\cref{subsec:cex-relative-weakening}). However, this does not guarantee that $\phi'$ holds on the model itself, and so we need to repeat the process, finding a new counterexample trace for $\phi'$ and weakening the interval again, driven by the iterative process. We finish once we produce a $\phi'$ that holds on the model.

\subsection{Iterative Weakening}\label{subsec:iterative-weakening}

Assume that we have an \MTL{} formula $\phi$ that has a temporal subformula $\psi\bintemporal{I}\psi'$ with an interval $I$ that we want to weaken. $\phi$ can be split into $\psi\bintemporal{I}\psi'$ and the surrounding \MTL{} context $C$, such that $\phi$ is logically equivalent to $C[\psi\bintemporal{I}\psi']$. Assume we also have a transition system $\model{}$. Using a model checker, we check whether $\phi$ holds on $\model{}$. If it holds then we are done, but if not, we will receive a counterexample trace $\pi$ through $\model{}$ for which $\pi\nvDash C[\psi\bintemporal{I}\psi']$. We can then weaken on this counterexample with
\begin{equation}
    I'=\weakenaux(C,\psi\bintemporal{I}\psi',\pi,0)
\end{equation}
which is described in \cref{subsec:cex-relative-weakening}. By \cref{theorem:weakening-algorithm-correctness}, if $I'=None$, then there exists no weakening $I''$ of $I$ such that $\pi\vDash C[\psi\bintemporal{I''}\psi']$, and so the same holds for the model $\model{}$. Otherwise, $I'$ is an interval such that $\pi\vDash C[\psi\bintemporal{I'}\psi']$. However, $C[\psi\bintemporal{I'}\psi']$ does not necessarily hold on $\model{}$, and so we model check again, creating an iterative loop that ends when we either produce an interval $I''$ that is a weakening of $I$ such that $C[\psi\bintemporal{I''}\psi']$ holds on $\model{}$, or show that no such weakening exists.

\subsection{Weakening on a Counterexample}\label{subsec:cex-relative-weakening}

Model checkers generally produce a specific type of infinite counterexample trace, called a \emph{lasso trace}.

\begin{definition}[Lasso traces]
A trace $\pi$ is \emph{lasso} if it can be separated into a finite prefix $\pipre$ and an infinitely repeating finite suffix $\pisuf$, forming
\begin{equation}
\pi=\pipre(\pisuf)^\omega.
\end{equation}
This restricts us to a subset of infinite traces that can be finitely represented. The finite length of a lasso trace is then defined as $|\pi|=|\pipre| + |\pisuf|$. Both $\pipre$ and $\pisuf$ must be minimal.
\end{definition}
We show that we can prove properties of an entire infinite lasso trace using only a finite \emph{covering interval}. Without this, we may need to iterate over the entire infinite trace, impacting completeness.
\begin{definition}[Covering intervals]
The suffix-covering interval of $\pi$, defined with respect to an interval $[a,b]$, is
\begin{equation}
\covering_\pi([a,b]) = [a, \min(b,\rightidx_\pi(a))]
\end{equation}
where we specify a finite \emph{end} of the infinite trace with
\begin{equation}
\rightidx_\pi(a) =
\begin{cases*}
  |\pi|        & if $a<|\pipre|$ \\
  a+|\pisuf|-1 & otherwise.
\end{cases*}
\end{equation}
\end{definition}

\begin{lemma}[Lasso trace coverage]\label{lemma:lasso-trace-coverage}
    Let $\phi$ be an \MTL{} formula, $\pi$ be a lasso trace, and $a\in\N$. If for all $t\in[a,\rightidx_\pi(a)]$ we have $\sat{\pi}{t}\phi$, then for all $t'\in\N$ with $t'\geq a$ we have $\sat{\pi}{t'}\phi$.
\end{lemma}

\begin{proof}
    We assume that for all $t\in[a,\rightidx_\pi(a)]$ we have $\sat{\pi}{t}\phi$, and, taking an arbitrary $t'\in\N$ with $t'\geq a$ want to prove that $\sat{\pi}{t'}\phi$. There are two cases. Firstly, if $t'<|\pi|$ then we know that this is within the $[a,\rightidx_\pi(a)]$ range and so we have $\sat{\pi}{t'}\phi$. Otherwise, if $t'\geq|\pi|$, we split $\pi$ into its prefix $\pipre$ and infinitely repeating suffix $\pisuf$, and want to prove that $\sat{\pipre(\pisuf)^\omega}{t'}\phi$. As $t'\geq|\pi|$, we can split it into $t'=|\pipre| + n\cdot|\pisuf| + m$ for some $n,m\in\N$ with $n\geq 1$ and $m<|\pisuf|$.
    \begin{align}
    \begin{split}
    & \sat{\pipre(\pisuf)^\omega}{|\pipre| + n\cdot|\pisuf| + m}\phi\\
    \implies& \sat{(\pisuf)^\omega}{n\cdot|\pisuf| + m}\phi\\
    \implies& \sat{(\pisuf)^\omega}{m}\phi\\
    \implies& \sat{\pipre(\pisuf)^\omega}{|\pipre| + m}\phi
    \end{split}
    \end{align}
    We know that $|\pipre| + m$ is in the $[a,\rightidx_\pi(a)]$ interval, so we have $\sat{\pi}{t'}\phi$.
    \qed
\end{proof}

We also define the \emph{optimality} of weakenings, used to prove that \cegiw{} will not produce an interval weakening that is any weaker than it needs to be.

\begin{algorithm}[t]
\SetAlgoNlRelativeSize{-1}
\caption{Weakening within a context $C$; non-temporal cases elided}
\label{alg:aux}

\Function{\weakenaux($C$, $\psi\bintemporal{\Iorig}\psi'$, $\pi$, $t$)}{
    \uIf{$C = [-]$}{
        \uIf{$\triangle = \mathcal{U}$}{
            \Return \weakenuntildirect$(\psi, \psi', \Iorig, \pi, t)$
        }
        \uElse(\tcp*[h]{$\triangle = \mathcal{R}$}){
            \Return \weakenreleasedirect$(\psi, \psi', \Iorig, \pi, t)$
        }
    }
    $\cdots$\\
    \uElseIf{$C = C \until{J} \phi$}{
        \Return \weakenuntilleft$(C, \phi, J, \psi\bintemporal{\Iorig}\psi', \pi, t)$
    }
    \uElseIf{$C = \phi \until{J} C$}{
        \Return \weakenuntilright$(\phi, C, J, \psi\bintemporal{\Iorig}\psi', \pi, t)$
    }
    \uElseIf{$C = C \release{J} \phi$}{
        \Return \weakenreleaseleft$(C, \phi, J, \psi\bintemporal{\Iorig}\psi', \pi, t)$
    }
    \uElseIf{$C = \phi \release{J} C$}{
        \Return \weakenreleaseright$(\phi, C, J, \psi\bintemporal{\Iorig}\psi', \pi, t)$
    }
}
\end{algorithm}

\begin{definition}[Optimality of right-bound extensions and contractions]
    An interval $I'$ is an optimal right-bound extension (resp.\ contraction) of an interval $I$ with respect to a context $C$, \MTL{} formulae $\psi$ and $\psi'$, a temporal operator $\triangle\in\{\mathcal{U},\mathcal{R}\}$, trace $\pi$, and time-step $t$, if
    \begin{equation}
    \sat{\pi}{t}C[\psi\bintemporal{I'}\psi']
    \end{equation}
    and either (a) $I=I'$, or (b) there exists no strict right-bound contraction (resp.\ extension) $I''$ of $I'$ such that $\sat{\pi}{t}C[\psi\bintemporal{I''}\psi']$.
\end{definition}
The entrypoint of \cegiw{} is \cref{alg:aux}, which recurses following the inductive structure of the \MTL{} context grammar\footnote{Implementations of elided cases are in the public repository.}. The proof of correctness and optimality follows the same inductive structure, with base cases for directly weakening the intervals of $\until{I}$ and $\release{I}$ (\cref{lemma:weakening-algorithm-direct-until,lemma:weakening-algorithm-direct-release}), and inductive cases for weakening subformulae on either side of both operators (\cref{lemma:weakening-algorithm-inductive-until-left,lemma:weakening-algorithm-inductive-until-right,lemma:weakening-algorithm-inductive-release-left,lemma:weakening-algorithm-inductive-release-right}).

\begin{algorithm}[t]
\SetAlgoNlRelativeSize{-1}
\caption{Directly weakening interval of $\mathcal{U}$}
\label{alg:weaken-direct-until}
\Function{\weakenuntildirect($\psi_l$, $\psi_r$, $[a,b]$, $\pi$, $t$)}{
    \For{$i \gets a$ \KwTo $\rightidx_\pi(a)$}{
        \If{$\pi, t+i\vDash \psi_r$}{\label{line:weaken-direct-until-1}
            \Return $[a, \max(b,i)]$
        }
        \If{$\pi, t+i\nvDash \psi_l$}{\label{line:weaken-direct-until-2}
            \textbf{break}
        }
    }
    \Return $None$
}
\end{algorithm}


Intuitively, to weaken a $\mathcal{U}$ formula $\psi_l \until{I} \psi_r$ in \cref{alg:weaken-direct-until}, the algorithm considers how the interval can be adjusted so that the formula becomes satisfied. Starting from time $t$, if the formula does not hold under the original interval, the only admissible weakening is to extend the right bound, thereby allowing additional time for the right subformula $\psi_r$ to become true while the left subformula $\psi_l$ continues to hold. The algorithm therefore extends the right bound incrementally until either the $\mathcal{U}$ formula holds on the given trace or no further extension is possible. In the former case, it returns the smallest such extension, yielding an optimal weakening; in the latter case, it reports that no interval weakening exists.

\begin{lemma}[$\mathcal{U}$ base case]\label{lemma:weakening-algorithm-direct-until}
    Let $I$ be an interval, $\psi_l$ and $\psi_r$ \MTL{} formulae, and $\pi$ a lasso trace. Then, for all timepoints $t\in\N$ with
    \begin{equation}
    I'=\weakenuntildirect(\psi_l,\psi_r,I,\pi,t),
    \end{equation}
    \noindent either $I'$ is an optimal right-bound extension of $I$ such that $\sat{\pi}{t}\psi_l\until{I'}\psi_r$, or $I'=None$, in which case there exists no such interval.
\end{lemma}

\begin{proof}
    We take an arbitrary $t$. Our proof for \cref{alg:weaken-direct-until} uses the loop invariant that $\forall j\in[a,\rightidx_\pi(a)]$ with $j<i$ (where $I=[a,b]$), we have that $\nsat{\pi}{t+j}\psi_r$ and $\sat{\pi}{t+j}\psi_l$. On first entry to the loop there is no such $j$, so this is trivially true. On reaching the end of the loop body, we know that $\nsat{\pi}{t+i}\psi_r$ and $\sat{\pi}{t+i}\psi_l$, and so in combination with the loop invariant we know that $\forall j\in I$ where $j\leq i$, we have $\nsat{\pi}{t+j}\psi_r$ and $\sat{\pi}{t+j}\psi_l$. Thus, the loop invariant is preserved. Suppose at the start of iteration $i$ that the loop invariant holds. If $\sat{\pi}{t+j}\psi_r$ on \cref{line:weaken-direct-until-1} then we return $I'=[a,\max(b,i)]$. This is an optimal right-bound extension of $I$ and we have that $\sat{\pi}{t}\psi_l\until{I'}\psi_r$.
    
    If $None$ is returned, then either we broke out of the loop early because for some $i\in[a,\rightidx_\pi(a)]$ we have $\nsat{\pi}{t+i}\psi_l$ at \cref{line:weaken-direct-until-2}, or we ran the loop to completion. In the first case, we know that $\nsat{\pi}{t+i}\psi_r$ as this is checked before at \cref{line:weaken-direct-until-1}, and so combining with the loop invariant we know that $\psi_r$ never held up until $\psi_l$ stopped holding, and so there is no right-bound extension $I'$ for which $\sat{\pi}{t}\psi_l\until{I'}\psi_r$. In the second case, by the loop invariant we have that for all $i\in[a,\rightidx_\pi(a)]$ we have $\sat{\pi}{t+i}\psi_l$ and $\nsat{\pi}{t+i}\psi_r$. By \cref{lemma:lasso-trace-coverage} we then have the same for all $i\in\N$ with $i\geq a$, and so there exists no right-bound extension $I'$ of $I$ that satisfies $\sat{\pi}{t}\psi_l\until{I'}\psi_r$.
    \qed
\end{proof}

\begin{lemma}[$\mathcal{R}$ base case]\label{lemma:weakening-algorithm-direct-release}
    Let $I$ be an interval, $\psi_l$ and $\psi_r$ \MTL{} formulae, and $\pi$ a lasso trace. Then, for all timepoints $t\in\N$ with
    \begin{equation}
    I'=\weakenreleasedirect(\psi_l,\psi_r,I,\pi,t),
    \end{equation}
    \noindent either $I'$ is an optimal right-bound contraction of $I$ such that $\sat{\pi}{t}\psi_l\release{I'}\psi_r$, or $I'=None$, in which case there exists no such interval.
\end{lemma}

\begin{proof}
Full proof and pseudocode is available in our public repository.
\qed
\end{proof}

We prove the inductive cases with an \MTL{} context $C$, an interval $I$, \MTL{} formulae $\psi$ and $\psi'$, a temporal operator $\triangle\in\{\mathcal{U}, \mathcal{R}\}$, and a lasso trace $\pi$. We use the induction hypothesis $P(C)$, that for all timepoints $t\in\N$ with $I'=\weakenaux(C,\psi\bintemporal{I}\psi', \pi,t)$, if $I'$ is an interval then $\pi,t\vDash C[\psi\bintemporal{I}\psi']$, and
\begin{compactitem}
    \item[1.] If $\triangle=\mathcal{U}$, then $I'$ is an optimal right-bound extension of $I$;

    \item[2.] If $\triangle=\mathcal{R}$, then $I'$ is an optimal right-bound contraction of $I$.
\end{compactitem}
If $I'=None$, then there exists no such interval in each case.

Intuitively, when weakening within the left subformula of a $\mathcal{U}$ operator in \cref{alg:aux-until-left}, we must ensure that the left subformula holds at every relevant timestep until the right subformula becomes true. Starting from time $t$, the algorithm therefore examines each timestep $t+i$ within the original interval and determines the interval weakening required for the left subformula to hold on the given trace at that point. Because the left operand of $\mathcal{U}$ is interpreted universally over the interval, the overall weakening must be strong enough to satisfy all such requirements. The algorithm therefore selects the weakest interval that subsumes all interval weakenings computed for individual timesteps. If no such interval exists, or if weakening fails at any timestep, the algorithm reports that no valid weakening can be found.

\begin{algorithm}[t]
\SetAlgoNlRelativeSize{-1}
\caption{Weakening within $\mathcal{U}$ on the left}
\label{alg:aux-until-left}

\Function{\weakenuntilleft($C$, $\phi$, $[a,b]$, $\psi\bintemporal{\Iorig}\psi'$, $\pi$, $t$)}{
    $\bfin \gets \min(b, \rightidx_\pi(a))$ \\
    $intervals \gets [\;]$ \\
    \For{$i \gets a$ \KwTo $\bfin$}{
        \If{$\pi, t + i \vDash \phi$}{\label{line:weaken-until-left-1}
            \If{$i = a$}{
                \Return $\Iorig$
            }
            \Return interval in $intervals$ with maximal absolute difference to $\Iorig$
        }
        $I \gets \weakenaux(C, \psi\bintemporal{\Iorig}\psi', \pi, t + i)$\\
        \If{$I = None$}{\label{line:weaken-until-left-2}
            \Return $None$
        }
        append $I$ to $intervals$
    }
    \Return $None$
}
\end{algorithm}

\begin{lemma}[$\mathcal{U}$-left inductive case]\label{lemma:weakening-algorithm-inductive-until-left}
    Let $C$ be an \MTL{} context, $I$ and $J$ intervals, $\phi$, $\psi$, and $\psi'$ \MTL{} formulae, $\triangle\in\{\mathcal{U}, \mathcal{R}\}$ a temporal operator, and $\pi$ a lasso trace. If $P(C)$ holds, then so does $P(C\until{J}\phi)$.
\end{lemma}

\begin{proof}
    For \cref{alg:aux-until-left} we assume the inductive hypothesis $P(C)$ and want to prove $P(C\until{J}\phi)$. We take an arbitrary $t$ and distinguish two cases, according to whether $\triangle$ is $\mathcal{U}$ or $\mathcal{R}$. In either case, by the induction hypothesis each recursive call evaluates to either $None$ or an optimal interval $I'$ related to $I$ by the corresponding relation (right-bound extension or contraction respectively) such that $\sat{\pi}{t+t'} C[\psi\bintemporal{I'}\psi']$.
    
    We use the loop invariant that, for all $j\in\covering_\pi(J)$ with $j<i$, we have that $\nsat{\pi}{t+j}\phi$ and that $I'=\weakenaux(C,\psi\bintemporal{I}\psi',\pi,t+j)$ is an interval such that $\sat{\pi}{t+j} C[\psi\bintemporal{I'}\psi']$. On first entry to the loop there is no such $j$, so this is trivially true. On reaching the end of the loop body, we know that $\weakenaux(C,\psi\bintemporal{I}\psi',\pi,t+i)\neq None$ from \cref{line:weaken-until-left-2}, and so by the induction hypothesis the recursive call must have produced a suitable interval $I'$. As we also know that $\nsat{\pi}{t+i}\phi$ from \cref{line:weaken-until-left-1}, the loop invariant is thus preserved for $j\leq i$. Suppose at the start of iteration $i$ that the loop invariant holds. If $\weakenaux(C,\psi\bintemporal{I}\psi',\pi,t+i) = None$ at \cref{line:weaken-until-left-2} then by the induction hypothesis we know that there is no suitable interval $I''$ for which $\sat{\pi}{t+i} C[\psi\bintemporal{I''}\psi']$, and by the loop invariant that there is no $j<i$ for which $\sat{\pi}{t+j}\phi$. Thus, there is no suitable interval $I''$ for which $\sat{\pi}{t} (C\until{J}\phi)[\psi\bintemporal{I''}\psi']$. If $\sat{\pi}{t+j}\phi$ at \cref{line:weaken-until-left-1} then we split on whether it is our first iteration or not. If $i=a$ (where $J=[a,b]$) then we know that $\sat{\pi}{t+a}\phi$, and so any interval will work. We simply return the original interval $\Iorig$.

    Otherwise, by the loop invariant we know that for all $j\in\covering_\pi(J)$ with $j<i$ --- of which there must be at least one as $i>a$ --- we have an interval $I'$ such that $\sat{\pi}{t+j} C[\psi\bintemporal{I'}\psi']$. Applying \cref{lemma:extension-weakening-order} if $\triangle=\mathcal{U}$, or \cref{lemma:contraction-weakening-order} if $\triangle=\mathcal{R}$, we obtain a maximum interval $I''$ such that for all $j\in\covering_\pi(J)$ with $j<i$ we have $\sat{\pi}{t+j} C[\psi\bintemporal{I''}\psi']$. If $\covering_\pi(J)=J$ then we have
    \begin{equation}
    \sat{\pi}{t} (C\until{J}\phi)[\psi\bintemporal{I''}\psi'].
    \end{equation}
    Otherwise, if $\covering_\pi(J)=[a,\rightidx_\pi(a)]$, then by \cref{lemma:lasso-trace-coverage} for all $k\in\N$ with $k\geq a$ we have $\sat{\pi}{t+k} C[\psi\bintemporal{I''}\psi']$, and so the above holds here too.
    \qed
\end{proof}

\begin{lemma}[$\mathcal{U}$-right inductive case]\label{lemma:weakening-algorithm-inductive-until-right}
    Let $C$ be an \MTL{} context, $I$ and $J$ intervals, $\phi$, $\psi$, and $\psi'$ \MTL{} formulae, $\triangle\in\{\mathcal{U}, \mathcal{R}\}$ a temporal operator, and $\pi$ a lasso trace. If $P(C)$ holds, then so does $P(\phi\until{J}C)$.
\end{lemma}

\begin{lemma}[$\mathcal{R}$-left inductive case]\label{lemma:weakening-algorithm-inductive-release-left}
    Let $C$ be an \MTL{} context, $I$ and $J$ intervals, $\phi$, $\psi$, and $\psi'$ \MTL{} formulae, $\triangle\in\{\mathcal{U}, \mathcal{R}\}$ a temporal operator, and $\pi$ a lasso trace. If $P(C)$ holds, then so does $P(C\release{J}\phi)$.
\end{lemma}

\begin{lemma}[$\mathcal{R}$-right inductive case]\label{lemma:weakening-algorithm-inductive-release-right}
    Let $C$ be an \MTL{} context, $I$ and $J$ intervals, $\phi$, $\psi$, and $\psi'$ \MTL{} formulae, $\triangle\in\{\mathcal{U}, \mathcal{R}\}$ a temporal operator, and $\pi$ a lasso trace. If $P(C)$ holds, then so does $P(\phi\release{J}C)$.
\end{lemma}
Full proofs for \cref{lemma:weakening-algorithm-inductive-until-right,lemma:weakening-algorithm-inductive-release-left,lemma:weakening-algorithm-inductive-release-right} are available in our public repository. We use these supporting lemmas to prove the correctness and optimality of \cegiw{}.
\begin{theorem}[Correctness for weakening]\label{theorem:weakening-algorithm-correctness}
    Let $C$ be an \MTL{} context, $I$ and $J$ intervals, $\phi$, $\psi$, and $\psi'$ \MTL{} formulae, $\triangle\in\{\mathcal{U}, \mathcal{R}\}$ a temporal operator, $\pi$ a lasso trace, and $t\in\N$ be a timepoint. Let $I'=\weakenaux(C,\psi\bintemporal{I}\psi',\pi,t)$. If $I'$ is an interval then $\pi \vDash C[\psi\bintemporal{I'}\psi']$, and
    \begin{compactitem}
        \item[1.] If $\triangle=\mathcal{U}$, then $I'$ is an optimal right-bound extension of $I$;
        
        \item[2.] If $\triangle=\mathcal{R}$, then $I'$ is an optimal right-bound contraction of $I$.
    \end{compactitem}
    If $I'=None$, then there exists no such interval in each case.
\end{theorem}

\begin{proof}
    We use the same induction hypothesis $P(C)$ defined for the preceding inductive lemmas. The non-temporal inductive cases are omitted for brevity.
    \begin{description}
    \item \textsc{Base case $[-]$:} By \cref{lemma:weakening-algorithm-direct-until} if $\triangle=\mathcal{U}$, and \cref{lemma:weakening-algorithm-direct-release} if $\triangle=\mathcal{R}$, we have $P([-])$.
    \item \textsc{Inductive case $C\until{J}\phi$:} Assuming $P(C)$, by \cref{lemma:weakening-algorithm-inductive-until-left} we have $P(C\until{J}\phi)$.
    \item \textsc{Inductive case $\phi\until{J}C$:} Assuming $P(C)$, by \cref{lemma:weakening-algorithm-inductive-until-right} we have $P(\phi\until{J}C)$.
    \item \textsc{Inductive case $C\release{J}\phi$:} Assuming $P(C)$, by \cref{lemma:weakening-algorithm-inductive-release-left} we have $P(C\release{J}\phi)$.
    \item \textsc{Inductive case $\phi\release{J}C$:} Assuming $P(C)$, by \cref{lemma:weakening-algorithm-inductive-release-right} we have $P(\phi\release{J}C)$.
    \end{description}
    \qed
\end{proof}

The time complexity of \cref{alg:aux} is $O(|\pi|^{\texttt{td}(\phi)})$. where $\pi$ is the counterexample trace and $\texttt{td}(\phi)$ is the \emph{temporal depth} of the \MTL{} formula $\phi$, i.e.\ the maximum number of nested temporal operators along any path in the syntax tree. In practice, $\texttt{td}(\phi)$ is typically very small.

\section{Evaluation}\label{sec:evaluation}

We evaluate how \textit{effective} interval weakening is in understanding the temporal behaviour of specifications, and how applicable it is to real-world requirements. To this end, we investigate the following research questions:
\begin{itemize}
    \item[\textbf{RQ1:}] How can \cegiw{} be used to explore and diagnose timing margins in \MTL{} specifications during early design? (\cref{subsec:demonstration})
    \item[\textbf{RQ2:}] To what extent do existing real-world requirements provide practical targets for interval weakening, and are such weakenings meaningful in their application domains? (\cref{subsec:case-studies})
\end{itemize}
\textbf{Choosing a model checker.} There are no industrial-strength model checkers for \MTL{} with pointwise semantics~\cite{brihaye2018, akshay2025}, yet many efficient tools exist for linear temporal logic~\cite{cavada2014,holzmann1997} (\LTL{}). Thus, we translate \MTL{} formulae into \LTL{} using the \emph{next} ($X$) operator~\cite[Remark~5.15]{baier2008a} and use an \LTL{} model checker. During preliminary investigation, it was found that symbolic \LTL{} model checkers such as \nuXmv{}~\cite{cavada2014} and \spin{}~\cite{holzmann1997} typically generate minimal counterexample traces. Weakening intervals with these usually only increments or decrements the bound rather than modifying it by a larger amount, which increases the number of calls made to the model checker dramatically. Our implementation uses \nuXmv{} in bounded model checking (BMC) mode, producing multiple counterexamples for a specific bound length, finding the optimal interval for each of them and returning the weakest, making it more likely that we make fewer calls to the model checker, thus improving the algorithm's efficiency. For our implementation we require the user to choose the BMC bound; while theoretical completeness can be preserved as completeness thresholds do exist for BMC~\cite{clarke2004a}, choosing a suitable bound still requires experience.

\subsection{Demonstration of \cegiw{} (RQ1)}\label{subsec:demonstration}

To address \textbf{RQ1}, we use an example based on a model of a foraging robot swarm~\cite{liu2010}. Robots are located in an arena and do a random walk to find food, which they then carry back to their home. Here they recharge and then repeat their foraging task. We model a single robot with a state machine, depicted abstractly in \cref{fig:foraging-fsm}.
\begin{figure}[t]
\centering
\begin{tikzpicture}[node distance=1.5cm and 3.5cm, on grid, auto]
\node[state] (deposit) {\texttt{deposit}};
\node[state, above=of deposit] (movetohome) {\texttt{moveToHome}};
\node[state, right=of movetohome] (grabfood) {\texttt{grabFood}};
\node[state, right=of grabfood] (movetofood) {\texttt{moveToFood}};
\node[state, right=of deposit] (homing) {\texttt{homing}};
\node[state, right=of homing] (scanarena) {\texttt{scanArena}};
\node[state, below=of deposit] (resting) {\texttt{resting}};
\node[left=5em of resting] (inv1) {};
\node[state, right=of resting] (leavinghome) {\texttt{leavingHome}};
\node[state, right=of leavinghome] (randomwalk) {\texttt{randomWalk}};

\path[->]
    (inv1) edge (resting)
    (resting) edge (leavinghome)
    (leavinghome) edge (randomwalk)
    (randomwalk) edge (homing)
    (homing) edge (resting)
    (randomwalk) edge[bend right=75](movetofood)
    (movetofood) edge (homing)
    (movetofood) edge[bend right=20] (scanarena)
    (scanarena) edge[bend right=20] (movetofood)
    (scanarena) edge (randomwalk)
    (movetofood) edge (grabfood)
    (grabfood) edge (movetohome)
    (movetohome) edge (deposit)
    (deposit) edge (resting);
\end{tikzpicture}
\caption{Abstract state transition system for the robot's foraging behaviour. The robot begins in a resting state, then searches for, collects, and deposits food.}
\label{fig:foraging-fsm}
\end{figure}
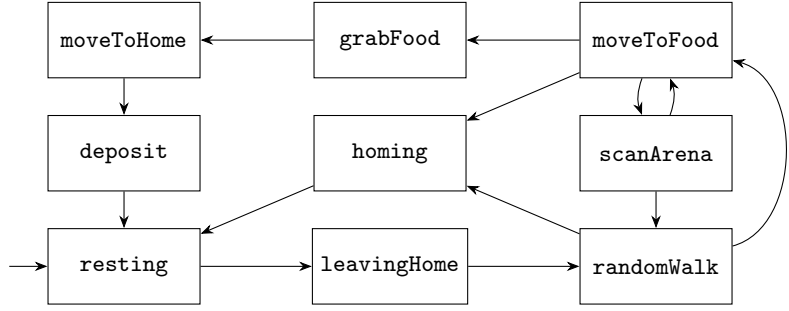
In the concrete transition system $\mathcal{M}$ the robot can remain in a given state for a configurable amount of time before it is forced to move to a next state. The transition system is specified concretely using \SMV{}, the language of the \nuXmv{} model checker~\cite{cavada2014}. We would like to prove that, after leaving the \texttt{resting} state, the robot will return to \texttt{resting} in at most 3 time units, represented by
\begin{equation}\label{eqn:max-return-to-resting-mtl}
\mathcal{M}\vDash\Box(\texttt{resting} \rightarrow \eventually{[1,3]}\texttt{resting})
\end{equation}
%
and translated from \MTL{} to \LTL{} as
\begin{equation}\label{eqn:max-return-to-resting-ltl}
\mathcal{M}\vDash\Box(\texttt{resting} \rightarrow X(\texttt{resting} \lor X(\texttt{resting} \lor X(\texttt{resting})))).
\end{equation}
%
If this does not hold in the transition system, we would like to weaken the interval to produce a new, weaker property that does hold. Using \cegiw{}, we find in the first iteration that no suitable weakening of the interval exists based on the counterexample in \cref{fig:min-to-resting-cex}, which shows an infinite loop between the \texttt{scanArena} and \texttt{moveToFood} states. This suggests a mistake in the modelling of the system, as in the real world the robot's battery would run out of charge.
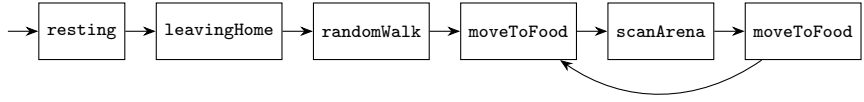
\begin{figure}[t]
\centering
\begin{tikzpicture}[auto,node distance=4mm]
\node[cex] (1) {\scriptsize\texttt{resting}};
\node[left=of 1] (inv1) {};
\node[cex, right=of 1] (2) {\scriptsize\texttt{leavingHome}};
\node[cex, right=of 2] (3) {\scriptsize\texttt{randomWalk}};
\node[cex, right=of 3] (4) {\scriptsize\texttt{moveToFood}};
\node[cex, right=of 4] (5) {\scriptsize\texttt{scanArena}};
\node[cex, right=of 5] (6) {\scriptsize\texttt{moveToFood}};
\node[right=of 6] (inv2) {};

\path[->]
    (inv1) edge (1)
    (1) edge (2)
    (2) edge (3)
    (3) edge (4)
    (4) edge (5)
    (5) edge (6)
    (6) edge[bend right=325] (4);
\end{tikzpicture}
\caption{Infinite lasso counterexample trace for the property in \cref{eqn:max-return-to-resting-mtl}.}
\label{fig:min-to-resting-cex}
\end{figure}
We amend the design by including in the requirements the notion of a battery that decreases as transitions are taken. While we are in \texttt{randomWalk}, \texttt{scanArena}, or \texttt{moveToFood} --- in other words, searching for food --- we monitor the battery level, and if it decreases below a certain threshold we abort and return home to recharge. The modified state transition system is depicted in \cref{fig:foraging-fsm-modified}.
\begin{figure}[t]
\centering
\begin{tikzpicture}[node distance=1.5cm and 3.5cm, on grid, auto]
\draw[thick,dashed] ($(movetofood.north west)+(-0.3,0.3)$) rectangle ($(randomwalk.south east)+(0.3,-0.3)$) node [midway, above=2.25] {\texttt{battery monitor}};
\draw[->,thick,dashed] ($($($(movetofood.north west)+(-0.3,0)$)!0.5!($(randomwalk.south west)+(-0.3,0)$)$)$) -- (homing) node [midway, below] {\texttt{abort}};
\node[state] (deposit) {\texttt{deposit}};
\node[state, above=of deposit] (movetohome) {\texttt{moveToHome}};
\node[state, right=of movetohome] (grabfood) {\texttt{grabFood}};
\node[state, right=of grabfood] (movetofood) {\texttt{moveToFood}};
\node[state, right=of deposit] (homing) {\texttt{homing}};
\node[state, right=of homing] (scanarena) {\texttt{scanArena}};
\node[state, below=of deposit] (resting) {\texttt{resting}};
\node[left=5em of resting] (inv1) {};
\node[state, right=of resting] (leavinghome) {\texttt{leavingHome}};
\node[state, right=of leavinghome] (randomwalk) {\texttt{randomWalk}};

\path[->]
    (inv1) edge (resting)
    (resting) edge (leavinghome)
    (leavinghome) edge (randomwalk)
    (randomwalk) edge (homing)
    (homing) edge (resting)
    (randomwalk) edge[bend right=75](movetofood)
    (movetofood) edge (homing)
    (movetofood) edge[bend right=20] (scanarena)
    (scanarena) edge[bend right=20] (movetofood)
    (scanarena) edge (randomwalk)
    (movetofood) edge (grabfood)
    (grabfood) edge (movetohome)
    (movetohome) edge (deposit)
    (deposit) edge (resting);
\end{tikzpicture}
\caption{Abstract state transition system for the robot's modified foraging behaviour. The battery monitor is represented by the dashed section on the right.}
\label{fig:foraging-fsm-modified}
\end{figure}
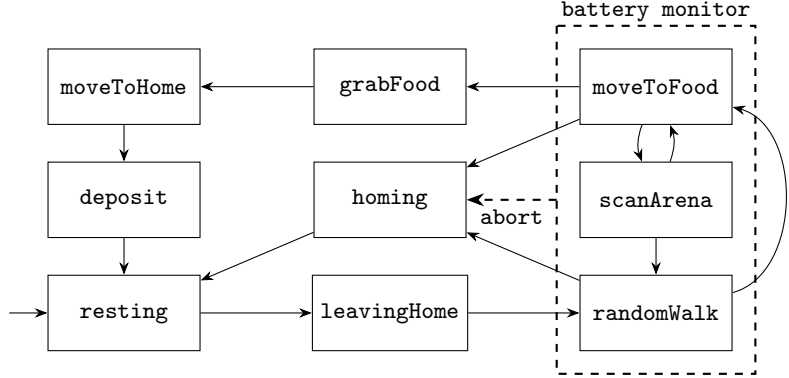
\begin{figure}[t]
\centering
\begin{subfigure}[t]{0.5\textwidth}
    \centering
    \begin{tikzpicture}
    \begin{axis}[
        xbar,
        xmin=1,
        width=5.75cm,
        height=3.5cm,
        xlabel={Interval},
        ylabel={Iteration},
        ytick=data,
        xtick distance=2,
        minor x tick num=1,
        x tick label as interval=false,
        y dir=reverse,
        ymin=-0.5,
        ymax=4.5,
        bar width=1.18em,
    ]
    \addplot+[black,fill=blue!80!white,opacity=0.5] coordinates {(20,4) (19,3) (17,2) (14,1) (3,0)};
    \node[
        anchor=west,
        font=\small,
    ] at (axis cs:14,4) {\textbf{valid}};
    \end{axis}
    \end{tikzpicture}
    \caption{Interval extension for \cref{eqn:max-return-to-resting-mtl}.}
    \label{fig:iterative-interval-extension}
\end{subfigure}%
\begin{subfigure}[t]{0.5\textwidth}
    \centering
    \begin{tikzpicture}
    \begin{axis}[
        xbar=1em,
        xmin=1,
        width=5.75cm,
        height=3.5cm,
        xlabel={Interval},
        ytick=data,
        xtick distance=2,
        minor x tick num=1,
        x tick label as interval=false,
        y dir=reverse,
        ymin=-0.5,
        ymax=1.5,
        bar width=2.95em,
    ]
    \addplot+[black,fill=blue!80!white,opacity=0.5] coordinates {(20,0) (3,1)};
    \node[
        anchor=west,
        font=\small,
    ] at (axis cs:3,1) {\textbf{valid}};
    \end{axis}
    \end{tikzpicture}
    \caption{Interval contraction for \cref{eqn:min-return-to-resting}.}
    \label{fig:iterative-interval-contraction}
\end{subfigure}
\caption{Iterative interval weakening to generate optimal, valid intervals.}
\end{figure}
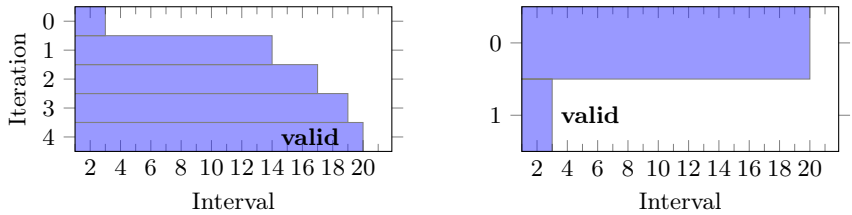
We check our desired property (\cref{eqn:max-return-to-resting-mtl}) against our amended model, and can see in \cref{fig:iterative-interval-extension} that in four iterations of \cegiw{} we extended the interval, and ended with the optimal interval which was then verified to hold in the system. So, the optimal property that holds in our amended system is
\begin{equation}\label{eqn:max-return-to-resting-corrected}
    \mathcal{M}\vDash\Box(\texttt{resting} \rightarrow \eventually{[1,20]}\texttt{resting}).
\end{equation}
Another property we are interested in is not the maximum time that the robot can spend away from home, but the \emph{minimum}. We wish the robot to spend at least 20 time units away from home, formalised as
\begin{equation}\label{eqn:min-return-to-resting}
\mathcal{M}\vDash\Box((\texttt{resting}\land\eventually{[1,1]}\neg\texttt{resting}) \rightarrow \generally{[1,20]}\neg\texttt{resting}).
\end{equation}
Again, we check this against our modified model and can see in \cref{fig:iterative-interval-contraction} that it took only one iteration to contract the interval, reaching the optimal interval which was then verified to hold in the system as
\begin{equation}\label{eqn:min-return-to-resting-corrected}
\mathcal{M}\vDash\Box((\texttt{resting}\land\eventually{[1,1]}\neg\texttt{resting}) \rightarrow \generally{[1,3]}\neg\texttt{resting}).
\end{equation}
By using \cegiw{}, we first identified that the original specification had a design flaw that allowed unwanted infinite loops. While a traditional model checker would conclude that the given specification does not hold, it cannot itself deduce that \emph{no} weakening exists. After modifying the specification, we then deduced both the maximum time that the robot can stay away from home, as well as the minimum time. \cegiw{} can thus provide value in both analysing existing requirements and supporting system modelling.

\subsection{Applicability of Interval Weakening to Real-World Requirements (RQ2)}\label{subsec:case-studies}
\begin{table}[t]
\caption{Interval-weakenable requirements in \fret{} case studies.}
\label{table:weakenable-reqs-case-studies}
\centering
\begin{tabular}{ |l|r|r| } 
\hline
\textbf{Case study} & \textbf{Total requirements} & \textbf{Weakenable requirements} \\\hline
Mechanical lung ventilator~\cite{farrell2024} & 121 & 57 \\\hline
Autonomous drone~\cite{sheridan2025} & 62 & 19 \\\hline
Lift-plus-cruise aircraft~\cite{pressburger2023} & 49 & 29 \\\hline
Aircraft engine controller~\cite{farrell2022} & 42 & 0 \\\hline
Inspection rover~\cite{bourbouh2021} & 15 & 1 \\\hline
Grasping for debris removal~\cite{farrell2022a} & 20 & 0 \\\hline
Robotic patterns~\cite{vazquez2024} & 36 & 11 \\\hline
LMCP challenges~\cite{mavridou2020} & 74 & 7 \\\hline
\textbf{Total} & \textbf{419} & \textbf{124} \\\hline
\end{tabular}
\end{table}
\begin{figure*}[t]
    \centering
    \begin{subfigure}[t]{0.475\textwidth}
        \texttt{\textcolor{fret_condition}{upon ControlLoopStart}
        \textcolor{fret_component}{System}
        \textcolor{fret_shall}{shall}
        \textcolor{fret_timing}{within 12 milliseconds}
        \textcolor{fret_response}{satisfy ControlLoopFinish}}
        \caption{Autonomous drone requirement REQ018 describing the maximum time the control loop can take to complete.}
        \label{fig:example-fret-requirement-REQ018}
    \end{subfigure}%
    \hspace{1em}%
    \begin{subfigure}[t]{0.475\textwidth}
        \texttt{\textcolor{fret_condition}{if powerFailure}
        \textcolor{fret_component}{System}
        \textcolor{fret_shall}{shall}
        \textcolor{fret_timing}{for 120 minutes}
        \textcolor{fret_response}{satisfy !off}}
        \vspace{\baselineskip}
        \caption{Mechanical lung ventilator requirement FUN37 describing how long the system must stay on after power failure.}
        \label{fig:example-fret-requirement-FUN37}
    \end{subfigure}%
    \caption{Example \fretish{} requirements from the case studies. Both are taken from systems that are fully implemented and operational in real-world settings.}
\end{figure*}
In this section, we analyse existing requirements from a number of real-world case studies to assess how often interval weakening is applicable, and how weakened timing bounds can be interpreted in their respective domains. These requirements are formalised using the Formal Requirements Elicitation Tool (\fret{})~\cite{giannakopoulou2020a} and are written in \fretish{}, a structured natural language that can be translated to \MTL{}~\cite{giannakopoulou2021}. \fretish{} requirements can have a \texttt{\textcolor{fret_timing}{timing}} field, on which we can use interval weakening to weaken the requirement itself. The number of requirements that can be weakened using interval weakening per case study is shown in \cref{table:weakenable-reqs-case-studies}. As an example, the requirement in \cref{fig:example-fret-requirement-REQ018} from the autonomous drone case study~\cite{sheridan2025} uses the \texttt{\textcolor{fret_timing}{within 12 milliseconds}} timing which specifies that if the condition holds in one state, then the consequent must hold within the next twelve states (assuming that state transitions correspond to a millisecond of time passing). The \MTL{} translation of this timing corresponds to the \MTL{} temporal operator $\Diamond_{[0,12]}$, and so the requirement corresponds to
%
%
\begin{equation}\label{eqn:example-fret-requirement-REQ018}
\Box(\texttt{\textcolor{fret_condition}{ControlLoopStart}}\rightarrow \eventually{[0,12]}(\texttt{\textcolor{fret_response}{ControlLoopFinish}})).
\end{equation}
Under system degradation, for example if the onboard communications network is degraded so that commands take longer to reach control surfaces, we may not be able to guarantee this and so would have to weaken the property by extending the interval, giving more time for the system to run its control loop, with an example weakening in
\begin{equation}\label{eqn:example-fret-requirement-REQ018-weakened}
\Box(\texttt{\textcolor{fret_condition}{ControlLoopStart}}\rightarrow \eventually{[0,24]}(\texttt{\textcolor{fret_response}{ControlLoopFinish}})).
\end{equation}
An example requirement from the mechanical lung ventilator case study~\cite{farrell2024} is shown in \cref{fig:example-fret-requirement-FUN37}, and as the \fretish{} timing \texttt{\textcolor{fret_timing}{for 120 minutes}} corresponds to the \MTL{} temporal operator $\Box_{[1,120]}$, the corresponding \MTL{} property is
\begin{equation}\label{eqn:example-fret-requirement-FUN37}
\Box(\texttt{\textcolor{fret_condition}{powerFailure}} \rightarrow  \generally{[1,120]}(\neg\texttt{\textcolor{fret_response}{off}})).
\end{equation}
This is a regulatory requirement~\cite{iso2023} and so, if it does not hold in the degraded system, it is critical to know by exactly how much it is violated. We may only be able to guarantee that the ventilator will stay on for at most 90 minutes after \texttt{\textcolor{fret_condition}{powerFailure}}, producing the weakening
\begin{equation}\label{eqn:example-fret-requirement-FUN37-weakened}
\Box(\texttt{\textcolor{fret_condition}{powerFailure}} \rightarrow  \generally{[1,90]}(\neg\texttt{\textcolor{fret_response}{off}})).
\end{equation}
Of the 127 interval-weakenable requirements in \cref{table:weakenable-reqs-case-studies}, 116 can be weakened by interval extension as in \cref{eqn:example-fret-requirement-REQ018-weakened}, and 11 by interval contraction as in \cref{eqn:example-fret-requirement-FUN37-weakened}.

Several case studies in \cref{table:weakenable-reqs-case-studies} have few or no requirements that can be weakened with interval weakening. These requirements are typically liveness properties specified with the \texttt{\textcolor{fret_timing}{eventually}} timing, which cannot be weakened further, or safety properties specified with the \texttt{\textcolor{fret_timing}{always}} timing, for which interval weakening would not be appropriate. For example, from the grasping for debris removal case study~\cite{farrell2022a},
\begin{equation}\label{eqn:example-fret-requirement-always}
\texttt{\textcolor{fret_component}{SV}
\textcolor{fret_shall}{shall}
\textcolor{fret_timing}{always}
\textcolor{fret_response}{satisfy !collide(SV, TGT)}}.
\end{equation}
To answer \textbf{RQ2}, we have shown that interval weakening is applicable to a substantial proportion of existing temporal requirements, and that such weakenings have meaningful interpretations in safety-critical domains. We have also answered \textbf{RQ1} by using \cegiw{} to identify problems in a specification, and then deduce useful timing properties in the fixed system.

\section{Conclusion}\label{sec:conclusion}

We present \cegiw{}, a novel algorithm for weakening intervals in \MTL{} properties of degraded systems, and prove its correctness and optimality. We demonstrate how \cegiw{} can be used during the design phase to understand system limitations under degradation, and explore how the formalised requirements of a number of real-world systems may be weakened against real implementations. This shows the applicability of \cegiw{} in the design of safety-critical systems for understanding the impacts of system degradation.

\textbf{Future work.}
A current limitation is that we only weaken on the right-hand-side of intervals, when both left- and right-bound modifications can produce valid weakenings. Restricting to only right-bound modifications creates a total order over the search space, so there is always a single optimum when multiple choices exist. Expanding to both left- and right-bound modifications creates a partial order over generated intervals, and so choosing between intervals is much less obvious. Future work will also explore other types of weakening, making syntactic changes to formulae beyond intervals.

Another limitation is that only a single interval can be weakened, and that the user must choose which one. Expanding to use multi-objective optimisation with Pareto optimality will improve the usability and applicability of this approach.

\textbf{Availability.} The implementation of CEGIW, full proofs, and all case study artefacts are available in our public repository\footnote{\url{https://github.com/benmandrew/CEGIW}}. Scripts are provided to reproduce all tables and examples reported in \cref{sec:evaluation}.

\printbibliography

@inproceedings{aarts2012,
  title = {Automata {{Learning}} through {{Counterexample Guided Abstraction Refinement}}},
  booktitle = {{{Formal Methods}}},
  author = {Aarts, Fides and Heidarian, Faranak and Kuppens, Harco and Olsen, Petur and Vaandrager, Frits},
  editor = {Giannakopoulou, Dimitra and M{\'e}ry, Dominique},
  year = 2012,
  pages = {10--27},
  publisher = {Springer},
  address = {Berlin, Heidelberg},
  %doi = {10.1007/978-3-642-32759-9_4},
  abstract = {Abstraction is the key when learning behavioral models of realistic systems. Hence, in most practical applications where automata learning is used to construct models of software components, researchers manually define abstractions which, depending on the history, map a large set of concrete events to a small set of abstract events that can be handled by automata learning tools. In this article, we show how such abstractions can be constructed fully automatically for a restricted class of extended finite state machines in which one can test for equality of data parameters, but no operations on data are allowed. Our approach uses counterexample-guided abstraction refinement: whenever the current abstraction is too coarse and induces nondeterministic behavior, the abstraction is refined automatically. Using Tomte, a prototype tool implementing our algorithm, we have succeeded to learn -- fully automatically -- models of several realistic software components, including the biometric passport and the SIP protocol.},
  isbn = {978-3-642-32759-9},
  langid = {english},
  keywords = {Automaton Learn,Equivalence Query,Input Symbol,Output Symbol,Session Initiation Protocol},
  file = {/home/y19056ba/Zotero/storage/QA3NB7IF/Aarts et al. - 2012 - Automata Learning through Counterexample Guided Abstraction Refinement.pdf}
}

@inproceedings{abreu2023a,
  title = {Exploring {{Automatic Specification Repair}} in {{Dafny Programs}}},
  booktitle = {{International Conference on Automated Software Engineering Workshops}},
  author = {Abreu, Alexandre and Macedo, Nuno and Mendes, Alexandra},
  year = 2023,
  month = sep,
  pages = {105--112},
  issn = {2151-0849},
 % doi = {10.1109/ASEW60602.2023.00019},
  urldate = {2024-11-27},
  abstract = {Formal verification has become increasingly crucial in ensuring the accurate and secure functioning of modern software systems. Given a specification of the desired behaviour, i.e. a contract, a program is considered to be correct when all possible executions guarantee the specification. Should the software fail to behave as expected, then a bug is present. Most existing research assumes that the bug is present in the implementation, but it is also often the case that the specified expectations are incorrect, meaning that it is the specification that must be repaired. Research and tools for providing alternative specifications that fix details missing during contract definition, considering that the implementation is correct, are scarce. This paper presents a preliminary tool, focused on Dafny programs, for automatic specification repair in contract programming. Given a Dafny program that fails to verify, the tool suggests corrections that repair the specification. Our approach is inspired by a technique previously proposed for another contract programming language and relies on Daikon for dynamic invariant inference. Although the tool is focused on Dafny, it makes use of specification repair techniques that are generally applicable to programming languages that support contracts. Such a tool can be valuable in various scenarios, such as when programmers have a reference implementation and need to analyse their contract options, or in educational contexts, where it can provide students with hints to correct their contracts. The results of the evaluation show that the approach is feasible in Dafny and that the overall process has reasonable performance but that there are stages of the process that need further imnrovements.},
  keywords = {automatic program repair,Computer bugs,Computer languages,Conferences,contract programming,contract repair,Contracts,Dafny,Dynamic programming,Maintenance engineering,Software systems},
  file = {/home/y19056ba/Zotero/storage/HFXX45XC/Abreu et al. - 2023 - Exploring Automatic Specification Repair in Dafny Programs.pdf;/home/y19056ba/Zotero/storage/LJZ2P47L/10298767.html}
}

@misc{akshay2025,
  title = {Efficient {{Verification}} of {{Metric Temporal Properties}} with {{Past}} in {{Pointwise Semantics}}},
  author = {Akshay, S. and Contractor, Prerak and Gastin, Paul and Govind, R. and Srivathsan, B.},
  year = 2025,
  month = oct,
  journal = {arXiv.org},
  urldate = {2025-11-24},
  howpublished = {https://arxiv.org/abs/2510.14699v1},
  langid = {english}
}

@article{alur1993,
  title = {Real-{{Time Logics}}: {{Complexity}} and {{Expressiveness}}},
  shorttitle = {Real-{{Time Logics}}},
  author = {Alur, Rajeev and Henzinger, Thomas A.},
  year = 1993,
  month = may,
  journal = {Information and Computation},
  volume = {104},
  number = {1},
  pages = {35--77},
  issn = {0890-5401},
 % doi = {10.1006/inco.1993.1025},
  urldate = {2025-11-18},
  abstract = {The theory of the natural numbers with linear order and monadic predicates underlies propositional linear temporal logic. To study temporal logics that are suitable for reasoning about real-time systems, we combine this classical theory of infinite state sequences with a theory of discrete time, via a monotonic function that maps every state to its time. The resulting theory of timed state sequences is shown to be decidable, albeit nonelementary, and its expressive power is characterized by {$\omega$}-regular sets. Several more expressive variants are proved to be highly undecidable. This framework allows us to classify a wide variety of real-time logics according to their complexity and expressiveness. Indeed, it follows that most formalisms proposed in the literature cannot be decided. We are, however, able to identify two elementary real-time temporal logics as expressively complete fragments of the theory of timed state sequences, and we present tableau-based decision procedures for checking validity. Consequently, these two formalisms are well-suited for the specification and verification of real-time systems.},
  file = {/home/y19056ba/Zotero/storage/T99UTBF6/Alur and Henzinger - 1993 - Real-Time Logics Complexity and Expressiveness.pdf;/home/y19056ba/Zotero/storage/HRRLXPLT/S0890540183710254.html}
}

@inproceedings{alur2013,
  title = {Syntax-Guided Synthesis},
  booktitle = {{{Formal Methods}} in {{Computer-Aided Design}}},
  author = {Alur, Rajeev and Bodik, Rastislav and Juniwal, Garvit and Martin, Milo M. K. and Raghothaman, Mukund and Seshia, Sanjit A. and Singh, Rishabh and {Solar-Lezama}, Armando and Torlak, Emina and Udupa, Abhishek},
  year = 2013,
  month = oct,
  pages = {1--8},
 % doi = {10.1109/FMCAD.2013.6679385},
  urldate = {2025-01-10},
  abstract = {The classical formulation of the program-synthesis problem is to find a program that meets a correctness specification given as a logical formula. Recent work on program synthesis and program optimization illustrates many potential benefits of allowing the user to supplement the logical specification with a syntactic template that constrains the space of allowed implementations. Our goal is to identify the core computational problem common to these proposals in a logical framework. The input to the syntax-guided synthesis problem (SyGuS) consists of a background theory, a semantic correctness specification for the desired program given by a logical formula, and a syntactic set of candidate implementations given by a grammar. The computational problem then is to find an implementation from the set of candidate expressions so that it satisfies the specification in the given theory. We describe three different instantiations of the counter-example-guided-inductive-synthesis (CEGIS) strategy for solving the synthesis problem, report on prototype implementations, and present experimental results on an initial set of benchmarks.},
  keywords = {Concrete,Grammar,Heuristic algorithms,Libraries,Production,Search problems,Syntactics},
  file = {/home/y19056ba/Zotero/storage/P77STNPH/Alur et al. - 2013 - Syntax-guided synthesis.pdf;/home/y19056ba/Zotero/storage/7TE7KAJ8/6679385.html}
}

@inproceedings{alur2013a,
  title = {Counter-Strategy Guided Refinement of {{GR}}(1) Temporal Logic Specifications},
  booktitle = {{{Formal Methods}} in {{Computer-Aided Design}}},
  author = {Alur, Rajeev and Moarref, Salar and Topcu, Ufuk},
  year = 2013,
  month = oct,
  pages = {26--33},
 % doi = {10.1109/FMCAD.2013.6679387},
  urldate = {2025-05-19},
  abstract = {The reactive synthesis problem is to find a finite-state controller that satisfies a given temporal-logic specification regardless of how its environment behaves. Developing a formal specification is a challenging and tedious task and initial specifications are often unrealizable. In many cases, the source of unrealizability is the lack of adequate assumptions on the environment of the system. In this paper, we consider the problem of automatically correcting an unrealizable specification given in the generalized reactivity (1) fragment of linear temporal logic by adding assumptions on the environment. When a temporal-logic specification is unrealizable, the synthesis algorithm computes a counter-strategy as a witness. Our algorithm then analyzes this counter-strategy and synthesizes a set of candidate environment assumptions that can be used to remove the counter-strategy from the environment's possible behaviors. We demonstrate the applicability of our approach with several case studies.},
  keywords = {Abstracts,Algorithm design and analysis,Educational institutions,Games,Polynomials,Transducers},
  file = {/home/y19056ba/Zotero/storage/H2N2MKYX/Alur et al. - 2013 - Counter-strategy guided refinement of GR(1) temporal logic specifications.pdf;/home/y19056ba/Zotero/storage/PLNN6MWV/6679387.html}
}

@inproceedings{andrew2026,
  title = {Weakening {{Goals}} in~{{Logical Specifications}}},
  booktitle = {Rigorous {{State-Based Methods}}},
  author = {Andrew, Ben M.},
  editor = {Leuschel, Michael and Ishikawa, Fuyuki},
  year = 2026,
  pages = {349--353},
  publisher = {Springer},
  address = {Cham},
 % doi = {10.1007/978-3-031-94533-5_22},
  abstract = {Logical specifications are widely used to represent software systems and their desired properties. Under system degradation or environmental changes, commonly seen in complex real-world robotic systems, these properties may no longer hold and so traditional verification methods will simply fail to construct a proof. However, weaker versions of these properties do still hold and can be useful for understanding the system's behaviour in uncertain conditions, as well as aiding compositional verification. We present a counterexample-guided technique for iteratively weakening properties, apply it to propositional logic specifications, and discuss planned extensions to state-based representations.},
  isbn = {978-3-031-94533-5},
  langid = {english},
  file = {/home/y19056ba/Zotero/storage/BHN8ERH9/Andrew - 2026 - Weakening Goals in Logical Specifications.pdf}
}

@book{baier2008a,
  title = {Principles of {{Model Checking}}},
  author = {Baier, Christel and Katoen, Joost-Pieter},
  year = 2008,
  month = apr,
  publisher = {MIT Press},
  googlebooks = {5dvxCwAAQBAJ},
  isbn = {978-0-262-30403-0},
  langid = {english}
}

@inproceedings{bourbouh2021,
  title = {Integrating {{Formal Verification}} and {{Assurance}}: {{An Inspection Rover Case Study}}},
  shorttitle = {Integrating {{Formal Verification}} and {{Assurance}}},
  booktitle = {{{NASA Formal Methods}}},
  author = {Bourbouh, Hamza and Farrell, Marie and Mavridou, Anastasia and Sljivo, Irfan and Brat, Guillaume and Dennis, Louise A. and Fisher, Michael},
  editor = {Dutle, Aaron and Moscato, Mariano M. and Titolo, Laura and Mu{\~n}oz, C{\'e}sar A. and Perez, Ivan},
  year = 2021,
  %volume = {12673},
  pages = {53--71},
  publisher = {Springer International Publishing},
  address = {Cham},
  %doi = {10.1007/978-3-030-76384-8_4},
  urldate = {2024-10-23},
  abstract = {The complexity and flexibility of autonomous robotic systems necessitates a range of distinct verification tools. This presents new challenges not only for design verification but also for assurance approaches. Combining the distinct formal verification tools, while maintaining sufficient formal coherence to provide compelling assurance evidence is difficult, often being abandoned for less formal approaches. In this paper we demonstrate, through a case study, how a variety of distinct formal techniques can be brought together in order to develop a justifiable assurance case. We use the AdvoCATE assurance case tool to guide our analyses and to integrate the artifacts from the formal methods that we use, namely: fret, cocosim and Event-B. While we present our methodology as applied to a specific Inspection Rover case study, we believe that this combination provides benefits in maintaining coherent formal links across development and assurance processes for a wide range of autonomous robotic systems.},
  isbn = {978-3-030-76383-1 978-3-030-76384-8},
  langid = {english},
  file = {/home/y19056ba/Zotero/storage/9AQCEPTK/Bourbouh et al. - 2021 - Integrating Formal Verification and Assurance An Inspection Rover Case Study.pdf}
}

@inproceedings{brihaye2018,
  title = {Efficient {{Algorithms}} and {{Tools}} for {{MITL Model-Checking}} and {{Synthesis}}},
  booktitle = {{International Conference on Engineering of Complex Computer Systems}},
  author = {Brihaye, Thomas and Geeraerts, Gilles and Ho, Hsi-Ming and Milchior, Arthur and Monmege, Benjamin},
  year = 2018,
  month = dec,
  pages = {180--184},
  %doi = {10.1109/ICECCS2018.2018.00027},
  urldate = {2025-11-24},
  abstract = {Metric Interval Temporal Logic (MITL) is an extension of the classical Linear Time Logic (LTL) that can be used to characterise real-time properties of computer systems. While the practical interest of MITL is undeniable, there is still today a remarkable lack of tool support for this logic. In this short paper, we report on our on-going work effort to complete the theoretical knowledge about MITL. We also report on our recently introduced tool MightyL, which translates MITL formulae into timed automata, enabling efficient model-checking of this logic. Finally, we sketch the future directions of our current line of research, which will be to extend MightyL to support reactive synthesis of MITL properties.},
  keywords = {Automata,Clocks,Computational modeling,Cost accounting,MITL,Model-checking,Real-time systems,Semantics,Synthesis,Timed automata,Tools},
  file = {/home/y19056ba/Zotero/storage/PIUAIP2Y/Brihaye et al. - 2018 - Efficient Algorithms and Tools for MITL Model-Checking and Synthesis.pdf;/home/y19056ba/Zotero/storage/U8M5G95L/8595072.html}
}

@inproceedings{brizzio2023,
  title = {Automated {{Repair}} of {{Unrealisable LTL Specifications Guided}} by {{Model Counting}}},
  booktitle = {{{Genetic}} and {{Evolutionary Computation Conference}}},
  author = {Brizzio, Mat{\'i}as and Cordy, Maxime and Papadakis, Mike and S{\'a}nchez, C{\'e}sar and Aguirre, Nazareno and Degiovanni, Renzo},
  year = 2023,
  month = jul,
  pages = {1499--1507},
  publisher = {Association for Computing Machinery},
  address = {New York, NY, USA},
  %doi = {10.1145/3583131.3590454},
  urldate = {2025-05-24},
  abstract = {The reactive synthesis problem consists of automatically producing correct-by-construction operational models of systems from high-level formal specifications of their behaviours. However, specifications are often unrealisable, meaning that no system can be synthesised from the specification. To deal with this problem, we present AuRUS, a search-based approach to repair unrealisable Linear-Time Temporal Logic (LTL) specifications. AuRUS aims at generating solutions that are similar to the original specifications by using the notions of syntactic and semantic similarities. Intuitively, the syntactic similarity measures the text similarity between the specifications, while the semantic similarity measures the number of behaviours preserved/removed by the candidate repair. We propose a new heuristic based on model counting to approximate semantic similarity. We empirically assess AuRUS on many unrealisable specifications taken from different benchmarks and show that it can successfully repair all of them. Also, compared to related techniques, AuRUS can produce many unique solutions while showing more scalability.},
  isbn = {979-8-4007-0119-1},
  file = {/home/y19056ba/Zotero/storage/3UP4BE2E/Brizzio et al. - 2023 - Automated Repair of Unrealisable LTL Specifications Guided by Model Counting.pdf}
}

@inproceedings{cavada2014,
  title = {The {{\textsc{nuXmv} Symbolic Model Checker}}},
  booktitle = {Computer {{Aided Verification}}},
  author = {Cavada, Roberto and Cimatti, Alessandro and Dorigatti, Michele and Griggio, Alberto and Mariotti, Alessandro and Micheli, Andrea and Mover, Sergio and Roveri, Marco and Tonetta, Stefano},
  editor = {Biere, Armin and Bloem, Roderick},
  year = 2014,
  pages = {334--342},
  publisher = {Springer International Publishing},
  address = {Cham},
 % doi = {10.1007/978-3-319-08867-9_22},
  abstract = {This paper describes the nuXmv symbolic model checker for finite- and infinite-state synchronous transition systems. nuXmv is the evolution of the nuXmv open source model checker. It builds on and extends nuXmv along two main directions. For finite-state systems it complements the basic verification techniques of nuXmv with state-of-the-art verification algorithms. For infinite-state systems, it extends the nuXmv language with new data types, namely Integers and Reals, and it provides advanced SMT-based model checking techniques.},
  isbn = {978-3-319-08867-9},
  langid = {english},
  keywords = {Model Check,Model Check Problem,Predicate Abstraction,Software Model Check,Symbolic Model Checker},
  file = {/home/y19056ba/Zotero/storage/NH4Z8UQL/Cavada et al. - 2014 - The nuXmv Symbolic Model Checker.pdf}
}

@inproceedings{cerqueira2022,
  title = {Timely {{Specification Repair}} for {{Alloy}} 6},
  booktitle = {Software {{Engineering}} and {{Formal Methods}}},
  author = {Cerqueira, Jorge and Cunha, Alcino and Macedo, Nuno},
  editor = {Schlingloff, Bernd-Holger and Chai, Ming},
  year = 2022,
  pages = {288--303},
  publisher = {Springer International Publishing},
  address = {Cham},
  %doi = {10.1007/978-3-031-17108-6_18},
  abstract = {This paper proposes the first mutation-based technique for the repair of Alloy 6 first-order temporal logic specifications. This technique was developed with the educational context in mind, in particular, to repair submissions for specification challenges, as allowed, for example, in the Alloy4Fun web-platform. Given an oracle and an incorrect submission, the proposed technique searches for syntactic mutations that lead to a correct specification, using previous counterexamples to quickly prune the search space, thus enabling timely feedback to students. Evaluation shows that, not only is the technique feasible for repairing temporal logic specifications, but also outperforms existing techniques for non-temporal Alloy specifications in the context of educational challenges.},
  isbn = {978-3-031-17108-6},
  langid = {english},
  keywords = {Alloy,First-order temporal logic,Formal methods education,Specification repair},
  file = {/home/y19056ba/Zotero/storage/254RAXET/Cerqueira et al. - 2022 - Timely Specification Repair for Alloy 6.pdf}
}

@inproceedings{clarke2000,
  title = {Counterexample-{{Guided Abstraction Refinement}}},
  booktitle = {Computer {{Aided Verification}}},
  author = {Clarke, Edmund and Grumberg, Orna and Jha, Somesh and Lu, Yuan and Veith, Helmut},
  editor = {Emerson, E. Allen and Sistla, Aravinda Prasad},
  year = 2000,
  pages = {154--169},
  publisher = {Springer},
  address = {Berlin, Heidelberg},
  %doi = {10.1007/10722167_15},
  abstract = {We present an automatic iterative abstraction-refinement methodology in which the initial abstract model is generated by an automatic analysis of the control structures in the program to be verified. Abstract models may admit erroneous (or ``spurious'') counterexamples. We devise new symbolic techniques which analyze such counterexamples and refine the abstract model correspondingly. The refinement algorithm keeps the size of the abstract state space small due to the use of abstraction functions which distinguish many degrees of abstraction for each program variable. We describe an implementation of our methodology in NuSMV. Practical experiments including a large Fujitsu IP core design with about 500 latches and 10000 lines of SMV code confirm the effectiveness of our approach.},
  isbn = {978-3-540-45047-4},
  langid = {english},
  keywords = {Abstract Model,Atomic Formula,Kripke Structure,Localization Reduction,Model Check},
  file = {/home/y19056ba/Zotero/storage/RNMC5FTY/Clarke et al. - 2000 - Counterexample-Guided Abstraction Refinement.pdf}
}

@incollection{clarke2004a,
  title = {Completeness and {{Complexity}} of {{Bounded Model Checking}}},
  booktitle = {Verification, {{Model Checking}}, and {{Abstract Interpretation}}},
  author = {Clarke, Edmund and Kroening, Daniel and Ouaknine, Jo{\"e}l and Strichman, Ofer},
  editor = {Goos, Gerhard and Hartmanis, Juris and Van Leeuwen, Jan and Steffen, Bernhard and Levi, Giorgio},
  year = 2004,
 % volume = {2937},
  pages = {85--96},
  publisher = {Springer Berlin Heidelberg},
  address = {Berlin, Heidelberg},
  doi = {10.1007/978-3-540-24622-0_9},
  urldate = {2025-11-19},
  abstract = {For every finite model M and an LTL property {$\phi$}, there exists a number CT (the Completeness Threshold ) such that if there is no counterexample to {$\phi$} in M of length CT or less, then M \textbar = {$\phi$}. Finding this number, if it is sufficiently small, offers a practical method for making Bounded Model Checking complete. We describe how to compute an over-approximation to CT for a general LTL property using B\textasciidieresis uchi automata, following the Vardi-Wolper LTL model checking framework. Based on the value of CT , we prove that the complexity of standard SAT-based BMC is doubly exponential, and that consequently there is a complexity gap of an exponent between this procedure and standard LTL model checking. We discuss ways to bridge this gap.},
  isbn = {978-3-540-20803-7 978-3-540-24622-0},
  langid = {english},
  file = {/home/y19056ba/Zotero/storage/5BRZ2EPF/Clarke et al. - 2004 - Completeness and Complexity of Bounded Model Checking.pdf}
}

@incollection{cobleigh2003,
  title = {Learning {{Assumptions}} for {{Compositional Verification}}},
  booktitle = {Tools and {{Algorithms}} for the {{Construction}} and {{Analysis}} of {{Systems}}},
  author = {Cobleigh, Jamieson M. and Giannakopoulou, Dimitra and P{\u a}s{\u a}reanu, Corina S.},
  editor = {Goos, Gerhard and Hartmanis, Juris and Van Leeuwen, Jan and Garavel, Hubert and Hatcliff, John},
  year = 2003,
 % volume = {2619},
  pages = {331--346},
  publisher = {Springer Berlin Heidelberg},
  address = {Berlin, Heidelberg},
 % doi = {10.1007/3-540-36577-X_24},
  urldate = {2025-01-06},
  abstract = {Compositional verification is a promising approach toaddressing the state explosion problem associated wihh model checking. One composition\_ technique advocates proving properties of a system by checking properties of its components in an assume-guarantee ......... s\_le. However, the application of thistech\_ique is \_\_fficult because it involves non-trivial human input. This paper presents a novel framework for performing assume-guarantee reasoning in an incremental and fully automated fashion. To check a component against a property, our approach generates assumptions that the environment needs to satisfy for the property to hold. These assumptions are then discharged on the rest of the system. Assumptions are computed by a learning algorithm. They are initially approximate, but become \_adually more precise by means of counterexamples obtained by model checking the component and its environment, alternately. This iterative process may at any stage conclude that the property is either true or false in the system. We have implemented our approach in the LTSA tool and applied it to the analysis of a NASA system.},
  isbn = {978-3-540-00898-9 978-3-540-36577-8},
  langid = {english},
  file = {/home/y19056ba/Zotero/storage/F5DG3MRS/Cobleigh et al. - 2003 - Learning Assumptions for Compositional Verification.pdf}
}

@inproceedings{farrell2022,
  title = {{{FRETting About Requirements}}: {{Formalised Requirements}} for~an~{{Aircraft Engine Controller}}},
  shorttitle = {{{FRETting About Requirements}}},
  booktitle = {Requirements {{Engineering}}: {{Foundation}} for {{Software Quality}}},
  author = {Farrell, Marie and Luckcuck, Matt and Sheridan, Ois{\'i}n and Monahan, Rosemary},
  editor = {Gervasi, Vincenzo and Vogelsang, Andreas},
  year = 2022,
  pages = {96--111},
  publisher = {Springer International Publishing},
  address = {Cham},
  doi = {10.1007/978-3-030-98464-9_9},
  isbn = {978-3-030-98464-9},
  langid = {english}
}

@article{farrell2022a,
  title = {Formal {{Modelling}} and {{Runtime Verification}} of {{Autonomous Grasping}} for {{Active Debris Removal}}},
  author = {Farrell, Marie and Mavrakis, Nikos and Ferrando, Angelo and Dixon, Clare and Gao, Yang},
  year = 2022,
  month = jan,
  journal = {Frontiers in Robotics and AI},
  volume = {8},
  publisher = {Frontiers},
  issn = {2296-9144},
  doi = {10.3389/frobt.2021.639282},
  urldate = {2025-11-26},
  langid = {english}
}

@inproceedings{farrell2024,
  title = {{{FRETting}} and~{{Formal Modelling}}: {{A Mechanical Lung Ventilator}}},
  shorttitle = {{{FRETting}} and~{{Formal Modelling}}},
  booktitle = {Rigorous {{State-Based Methods}}},
  author = {Farrell, Marie and Luckcuck, Matt and Monahan, Rosemary and Reynolds, Conor and Sheridan, Ois{\'i}n},
  editor = {Bonfanti, Silvia and Gargantini, Angelo and Leuschel, Michael and Riccobene, Elvinia and Scandurra, Patrizia},
  year = 2024,
  pages = {360--383},
  publisher = {Springer Nature Switzerland},
  address = {Cham},
  doi = {10.1007/978-3-031-63790-2_28},
  isbn = {978-3-031-63790-2},
  langid = {english}
}

@inproceedings{gazzola2018,
  title = {Automatic Software Repair: A Survey},
  shorttitle = {Automatic Software Repair},
  booktitle = {{International Conference on Software Engineering}},
  author = {Gazzola, Luca and Micucci, Daniela and Mariani, Leonardo},
  year = 2018,
  month = may,
 % series = {{{ICSE}} '18},
  pages = {1219},
  publisher = {Association for Computing Machinery},
  address = {New York, NY, USA},
  doi = {10.1145/3180155.3182526},
  urldate = {2025-03-20},
  abstract = {Debugging software failures is still a painful, time consuming, and expensive process. For instance, recent studies showed that debugging activities often account for about 50\% of the overall development cost of software products [3]. There are many factors contributing to the cost of debugging, but the most impacting one is the extensive manual effort that is still required to identify and remove faults. So far, the automation of debugging activities essentially resulted in the development of techniques that provide useful insights about the possible locations of faults, the inputs and states of the application responsible for the failures, as well as the anomalous operations executed during failures. However, developers must still put a relevant effort on the analysis of the failed executions to exactly identify the faults that must be fixed. In addition, these techniques do not help the developers with the synthesis of an appropriate fix.},
  isbn = {978-1-4503-5638-1},
  file = {/home/y19056ba/Zotero/storage/U8K3Q4FP/Gazzola et al. - 2019 - Automatic Software Repair A Survey.pdf}
}

@article{giannakopoulou2020a,
  title = {Formal {{Requirements Elicitation}} with {{FRET}}},
  author = {Giannakopoulou, Dimitra and Pressburger, Thomas and Mavridou, Anastasia and Rhein, Julian and Schumann, Johann and Shi, Nija},
  year = 2020,
  journal = {International Working Conference on Requirements Engineering: Foundation for Software Quality},
  langid = {english}
}

@article{giannakopoulou2021,
  title = {Automated Formalization of Structured Natural Language Requirements},
  author = {Giannakopoulou, Dimitra and Pressburger, Thomas and Mavridou, Anastasia and Schumann, Johann},
  year = 2021,
  month = sep,
  journal = {Information and Software Technology},
  volume = {137},
  pages = {106590},
  issn = {0950-5849},
  doi = {10.1016/j.infsof.2021.106590},
  urldate = {2025-04-23}
}

@article{holzmann1997,
  title = {The Model Checker {{SPIN}}},
  author = {Holzmann, Gerard J.},
  year = 1997,
  month = may,
  journal = {IEEE Transactions on Software Engineering},
  volume = {23},
  number = {5},
  pages = {279--295},
  issn = {0098-5589, 1939-3520, 2326-3881},
  doi = {10.1109/32.588521},
  urldate = {2025-10-20},
  copyright = {https://ieeexplore.ieee.org/Xplorehelp/downloads/license-information/IEEE.html},
  langid = {english}
}

@inproceedings{howar2011,
  title = {Automata {{Learning}} with {{Automated Alphabet Abstraction Refinement}}},
  booktitle = {Verification, {{Model Checking}}, and {{Abstract Interpretation}}},
  author = {Howar, Falk and Steffen, Bernhard and Merten, Maik},
  editor = {Jhala, Ranjit and Schmidt, David},
  year = 2011,
  pages = {263--277},
  publisher = {Springer},
  address = {Berlin, Heidelberg},
  doi = {10.1007/978-3-642-18275-4_19},
  isbn = {978-3-642-18275-4},
  langid = {english}
}

@misc{iso2023,
  type = {80601-2-12},
  title = {Particular Requirements for Basic Safety and Essential Performance of Critical Care Ventilators},
  shorttitle = {{{ISO}} 80601-2-12},
  author = {{ISO}},
  year = 2023,
  number = {80601-2-12},
  urldate = {2025-12-01},
  langid = {english}
}

@article{koymans1990,
  title = {Specifying Real-Time Properties with Metric Temporal Logic},
  author = {Koymans, Ron},
  year = 1990,
  month = nov,
  journal = {Real-Time Systems},
  volume = {2},
  number = {4},
  pages = {255--299},
  issn = {1573-1383},
  doi = {10.1007/BF01995674},
  urldate = {2025-11-27},
  langid = {english}
}

@article{liu2010,
  title = {Modeling and {{Optimization}} of {{Adaptive Foraging}} in {{Swarm Robotic Systems}}},
  author = {Liu, Wenguo and Winfield, Alan F. T.},
  year = 2010,
  month = dec,
  journal = {The International Journal of Robotics Research},
  volume = {29},
  number = {14},
  pages = {1743--1760},
  publisher = {SAGE Publications Ltd STM},
  issn = {0278-3649},
  doi = {10.1177/0278364910375139},
  urldate = {2025-08-21},
  langid = {english}
}

@inproceedings{maoz2019,
  title = {Symbolic {{Repairs}} for {{GR}}(1) {{Specifications}}},
  booktitle = {{International Conference on Software Engineering}},
  author = {Maoz, Shahar and Ringert, Jan Oliver and Shalom, Rafi},
  year = 2019,
  month = may,
  pages = {1016--1026},
  issn = {1558-1225},
  doi = {10.1109/ICSE.2019.00106},
  urldate = {2025-05-23}
}

@inproceedings{mavridou2020,
  title = {The {{Ten Lockheed Martin Cyber-Physical Challenges}}: {{Formalized}}, {{Analyzed}}, and {{Explained}}},
  shorttitle = {The {{Ten Lockheed Martin Cyber-Physical Challenges}}},
  booktitle = {{International Requirements Engineering Conference}},
  author = {Mavridou, Anastasia and Bourbouh, Hamza and Giannakopoulou, Dimitra and Pressburger, Thomas and Hejase, Mohammad and Garoche, Pierre-Loïc and Schumann, Johann},
  year = 2020,
  month = aug,
  pages = {300--310},
  issn = {2332-6441},
  doi = {10.1109/RE48521.2020.00040},
  urldate = {2025-11-27}
}

@inproceedings{pressburger2023,
  title = {Authoring, {{Analyzing}}, and~{{Monitoring Requirements}} for~a~{{Lift-Plus-Cruise Aircraft}}},
  booktitle = {Requirements {{Engineering}}: {{Foundation}} for {{Software Quality}}},
  author = {Pressburger, Tom and Katis, Andreas and Dutle, Aaron and Mavridou, Anastasia},
  editor = {Ferrari, Alessio and Penzenstadler, Birgit},
  year = 2023,
  pages = {295--308},
  publisher = {Springer Nature Switzerland},
  address = {Cham},
  doi = {10.1007/978-3-031-29786-1_21},
  isbn = {978-3-031-29786-1},
  langid = {english}
}

@inproceedings{sheridan2025,
  title = {Sharper {{Specs}} for~{{Smarter Drones}}: {{Formalising Requirements}} with~{{FRET}}},
  shorttitle = {Sharper {{Specs}} for~{{Smarter Drones}}},
  booktitle = {Requirements {{Engineering}}: {{Foundation}} for {{Software Quality}}},
  author = {Sheridan, Ois{\'i}n and Becker, Leandro Buss and Farrell, Marie and Luckcuck, Matt and Monahan, Rosemary},
  editor = {Hess, Anne and Susi, Angelo},
  year = 2025,
  pages = {350--362},
  publisher = {Springer Nature Switzerland},
  address = {Cham},
  doi = {10.1007/978-3-031-88531-0_25},
  isbn = {978-3-031-88531-0},
  langid = {english}
}

@inproceedings{vazquez2024,
  title = {Robotics: {{A New Mission}} for~{{FRET Requirements}}},
  shorttitle = {Robotics},
  booktitle = {{{NASA Formal Methods}}},
  author = {V{\'a}zquez, Gricel and Mavridou, Anastasia and Farrell, Marie and Pressburger, Tom and Calinescu, Radu},
  editor = {Benz, Nathaniel and Gopinath, Divya and Shi, Nija},
  year = 2024,
  pages = {359--376},
  publisher = {Springer Nature Switzerland},
  address = {Cham},
  doi = {10.1007/978-3-031-60698-4_22},
  isbn = {978-3-031-60698-4},
  langid = {english}
}

\end{document}